\documentclass[a4paper,onecolumn,11pt, accepted=2025-07-21]{quantumarticle}
\pdfoutput=1
\usepackage[T1]{fontenc}

\usepackage{graphicx}
\usepackage{changes}
\usepackage{hyperref}
\usepackage{color}
\usepackage[numbers]{natbib}

\usepackage{thmtools, thm-restate}

\usepackage{amsmath,amssymb,amsthm}
\usepackage{bm}
\usepackage{mathtools}
\usepackage{dsfont}
\usepackage{braket}
\newcommand{\D}{\mathcal{D}}

\renewcommand{\L}{\mathcal{L}}
\renewcommand{\H}{\mathcal{H}}
\newcommand{\ketbra}[1]{\ket{#1}\!\bra{#1}}
\newcommand{\Tr}{\mathrm{Tr}}
\newcommand{\epsBetaPM}{\epsilon_{{}_{\text{BKP}}}}
\newcommand{\epsOrig}{\epsilon_{{}_{\text{orig}}}}

\newcommand{\piBetaPM}{\pi_{{}_{\text{BKP}}}}
\newcommand{\piOrig}{\pi_{{}_{\text{orig}}}}%P_guess
\newcommand{\Pim}{P_{\text{guess}}(A|Y\mathcal{Z}(Q))}
\newcommand{\Pgen}{P_\text{{guess}}(A|YZ\Phi_{t_{comp}}(Q_{\text{in}}))}

\newcommand{\abort}{P(\text{abort})}

\renewcommand{\emph}[1]{\textit{#1}}
\newcommand{\s}{\bm{s}}
\newcommand{\e}{(i,j)}
\renewcommand{\S}{\mathcal{S}}

\newcommand{\Mlist}{\bm{s}_M}
\newtheorem{definition}{Definition}[section]
\newtheorem{theorem}{Theorem}[section]
\newtheorem{lemma}{Lemma}[section]
\newtheorem{corollary}{Corollary}[section]

\theoremstyle{remark}
\newtheorem{remark}{Remark}[section]

\begin{document}

\title{Hybrid Quantum Cryptography from Communication Complexity}

\author{Francesco Mazzoncini} 
\affiliation{T\'el\'ecom Paris-LTCI, Institut Polytechnique de Paris, 19 Place Marguerite Perey, 91120 Palaiseau, France}
\affiliation{Orange Innovation, Orange Gardens, 44 avenue de la République, Châtillon,
France}
\author{Balthazar Bauer}
\affiliation{LMV, Université de Versailles – Saint-Quentin-en-Yvelines, 55 Avenue de Paris, 78646 Versailles, France}
\author{Peter Brown}
\affiliation{T\'el\'ecom Paris-LTCI, Institut Polytechnique de Paris, 19 Place Marguerite Perey, 91120 Palaiseau, France}
\author{Romain Alléaume}
\affiliation{T\'el\'ecom Paris-LTCI, Institut Polytechnique de Paris, 19 Place Marguerite Perey, 91120 Palaiseau, France}

\maketitle

\begin{abstract}
We introduce an explicit construction for a key distribution protocol in the Quantum Computational Timelock (QCT) security model, where one assumes that computationally secure encryption may only be broken after a time much longer than the coherence time of available quantum memories.
 Taking advantage of the QCT assumptions, we build a key distribution protocol called HM-QCT from the Hidden Matching problem for which there exists an exponential gap in one-way communication complexity between classical and quantum strategies.

We establish that the security of HM-QCT against arbitrary i.i.d.\ attacks can be reduced to the difficulty of solving the underlying Hidden Matching problem with classical information. Legitimate users, on the other hand, can use quantum communication, which gives them the possibility of sending multiple copies of the same quantum state while retaining an information advantage. 
This leads to an everlasting secure key distribution scheme over $n$ bosonic modes. Such a level of security is unattainable with purely classical techniques. Remarkably, the scheme remains secure with up to $\mathcal{O}\left( \frac{\sqrt{n}}{\log(n)}\right)$ input photons for each channel use, extending the functionalities and potentially outperforming QKD rates by several orders of magnitude.
% \keywords{Quantum Cryptography  \and Communication Complexity \and Information Theory.}
\end{abstract}
\section{Introduction}
\subsection{Quantum Cryptography}
Quantum cryptography has been largely defined \cite{bennettQuantumCryptographyPublic2014} as a novel form of cryptography that does not rely on computational hardness assumptions but on quantum means, and in particular quantum communications, to achieve information-theoretic security.
Encoding classical information redundantly, on multiple copies of the same quantum state, could be highly beneficial from an engineering viewpoint, allowing for higher rates and better resilience to loss. However, this is a problem for the security of many quantum cryptography protocols as it would allow the adversary to gain more information about the underlying state than if just a single copy is sent.
This limitation translates into a mean photon number that is typically upper bounded by 1 in QKD protocols, and more generally into the existence of a fundamental rate-loss trade-off that severely limits the distances over which we can perform QKD \cite{pirandolaFundamentalLimitsRepeaterless2017}.

In this work, we explore a new approach to quantum cryptography, by considering a hybrid security model. In particular, we unlock the possibility of sending multiple copies of the same state to perform key establishment with \emph{everlasting security} \cite{unruhEverlastingMultipartyComputation2013} with performances that go beyond standard QKD. We specifically consider a cryptographic protocol built on top of the Hidden Matching quantum communication complexity problem \cite{gavinskyExponentialSeparationsOneway2007,bar-yossefExponentialSeparationQuantum2004}, for which there exists an exponential gap between classical and quantum strategies. We prove its security by establishing a reduction to the classical strategies for this communication complexity problem, effectively connecting the field of communication complexity and quantum cryptography.

\subsection{QCT Security model}
\label{sec:intro_QCT}
A novel security model called \textit{Quantum Computational Time-lock (QCT}) was introduced in \cite{vyasEverlastingSecureKey2020}, building a bridge between the often disparate worlds of classical and quantum cryptography.
The model is based on two nested assumptions. The first one is that Alice and Bob can use a $t_{comp}$-secure encryption scheme.
\begin{definition}[$t_{comp}$-secure encryption scheme]
\label{def:encr_scheme}
An encryption scheme (\texttt{Gen}; \texttt{Enc}; \texttt{Dec}) is said to be $t_{comp}$-secure if it is computationally secure with respect to any unauthorized attacker Eve for a time at least $t_{comp}$, after a ciphertext is exchanged on the classical channel.
\end{definition}
The second assumption is that an adversary Eve cannot reliably store a quantum state during a time larger than $t_{comp}$ i.e. that she has access to what we call a $(t_{comp},\delta)$-noisy quantum memory, defined as follows.
\begin{definition} [$(t_{comp},\delta)$-noisy quantum memory]
\label{def:qmemory}
A $(t_{comp},\delta)$-noisy quantum memory is a Markovian time-dependent quantum memory $\Phi_{t}$ such that at time $t_{comp}$:
 \begin{equation}
 \label{assum:delta-memory}
\lVert \Phi_{t_{comp}} - \mathcal{F} \rVert_{\diamond} \le \delta \;,
\end{equation}
where $\mathcal{F}(\rho) \coloneqq \frac{\Tr[\rho]}{d_{out}}\mathds{1}_{d_{out}}\;,$ $\lVert \cdot \rVert_{\diamond}$ is the diamond norm \cite{watrousTheoryQuantumInformation2018a} and $d_{out}$ is the dimension of the output of the quantum memory.
\end{definition}
\noindent
In other words, a $(t_{comp},\delta)$-noisy quantum memory is a quantum memory that is hard to distinguish (parametrized by a parameter $\delta$) from a completely mixing channel $\mathcal{F}$, when it stores a quantum state for a time $t_{comp}$ or longer. One can note that assuming that the coherence time of available quantum memories is much shorter than $t_{comp}$ corresponds to taking $\delta \ll 1$.

 \subsubsection{Validity of QCT model}
 The validity of the QCT model is solidly grounded in practice when one considers existing and prospective quantum storage capabilities \cite{heshamiQuantumMemoriesEmerging2016} and puts them in perspective with an extremely conservative lower bound on the time $t_{comp}$ for which current encryption schemes would be considered secure, such as $t_{comp} \ge 10^{5} \, s \sim 1$ day. Moreover, it is interesting to understand that although the QCT assumptions set some limits to the scaling of quantum error-corrected quantum memory, it does not rule out the possibility of having useful quantum computers. Extrapolating for instance on \cite{gidneyHowFactor20482021} we see that 20 million noisy (with physical gate error $10^{-3}$) qubits would be sufficient to factor a RSA 2048 key, using $10^4$ logical qubits. However, considering the same resources, and the same number of logical qubits, they could be stored during only few hours. This would hence not rule out the conservative QCT assumptions mentioned above. 
%\\[0.2cm]
%\noindent
We should also stress that the QCT approach enables us to build key establishment schemes that offer everlasting security. This means that the secret keys can be provably secure against an adversary who is computationally unbounded after quantum storage decoherence, where the decoherence time to be considered is the one technologically available at the time of protocol execution. In particular, security holds against any future progress of the attacker's computational and quantum storage capabilities.
% To conclude, A first version of the QCT security model has been introduced in \cite{vyasEverlastingSecureKey2020} together with a key agreement protocol called MUB-Quantum Computational Timelock, where Alice encodes a single bit onto a qudit state selected from a set of $d+1$ mutually unbiased bases.  
\subsection{Our work}

\subsubsection{Sketch of the protocol} In our work we introduce an explicit construction for a new key distribution protocol called Hidden-Matching Quantum Computational Timelock (HM-QCT). It is built on top of a computational problem with a boolean output, called $\beta$-Partial Matching $(\beta PM)$ \cite{gavinskyExponentialSeparationsOneway2007}, for which $\Omega(\sqrt{n})$ bits of communication from Alice to Bob are required, against only $\mathcal{O}\left(\log(n)\right)$ qubits, with $n$ the length of input $x$. In each round of the HM-QCT protocol Alice generates both inputs $x$ and $y$ and shares the latter with Bob using a computationally secure encryption scheme. Alice and Bob can then solve the $\beta PM$ protocol with a quantum strategy to extract a bit, sending $m$ copies of the same $n$-dimensional quantum state.
See Figure \ref{fig:HM-QCT_protocol} for a pictorial representation.
Finally, by performing standard classical post-processing to their string of bits, they can distill a secure key.
\begin{figure}[htb!]
\vspace{0.4cm}
\centering
\includegraphics[width= 0.8\textwidth]{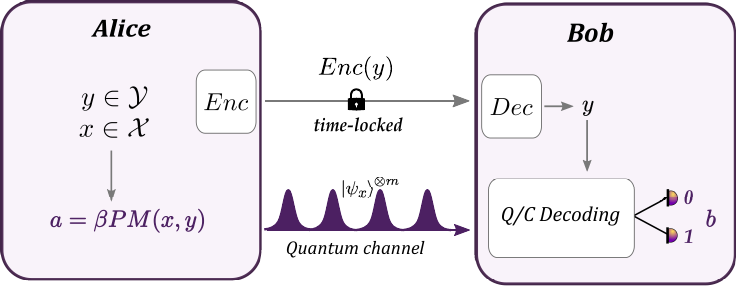}
\caption{One round of the HM-QCT protocol}
\label{fig:HM-QCT_protocol}
\end{figure}

%In our security analysis we considered an idealized case where one runs infinitely many rounds of the protocol and a possible adversary behaves independently and identically in each round. Within this setting,

 \subsubsection{Advantages over standard QKD}
Our results illustrate that the QCT hybrid security model constitutes a promising route to enhance the capabilities and effectiveness of quantum cryptography, while retaining some core advantage against classical cryptography: the possibility of providing everlasting security. In particular, our protocol offers the following benefits: 
 \begin{itemize}
\item \textbf{Boosted key rates:} Security can be achieved while sending up to $\mathcal{O}\left( \frac{\sqrt{n}}{\log(n)}\right)$ photons per channel use, overcoming the standard limit of one photon per channel use. As detailed in Section \ref{sec:key_estab}, the HM-QCT protocol, based on the $\beta PM$ problem, can be implemented with 2 single-mode threshold detectors, and performance can hence be benchmarked with 2-output-mode protocols. The fact that security can be achieved with many photons per channel use leads to asymptotic achievable secret key rates that can be boosted by a factor $\mathcal{O}\left( \frac{\sqrt{n}}{\log(n)}\right)$ with respect to BB84 QKD. As illustrated on Figure \ref{fig:keyrate_comparison}, HM-QCT could moreover overcome the fundamental secret key capacity \cite{pirandolaFundamentalLimitsRepeaterless2017}.
\item \textbf{Improved functionalities:} A notable advantage of enabling multiple copies per channel use is the potential to consistently hit the classical capacity of one bit per channel use over relatively short distances, as illustrated in Figure\ \ref{fig:keyrate_zoom} — a feat impossible in standard QKD. Moreover, multiple photons not only offer improved efficiency but also enable multicast key distribution (conference key agreement \cite{murtaQuantumConferenceKey2020})  with up to $\mathcal{O}\left(\frac{\sqrt{n}}{\log(n)}\right)$ authorized Bobs simultaneously.
\item \textbf{Security with untrusted measurement devices}: Eve's information can be upper bounded by considering only the dimension of the quantum state prepared by Alice and does not require (as in standard QKD) any information about Bob's measurement results. 
Additionally, since Alice can send multiple photons, HM-QCT can operate at short distances with a detection probability approaching 100\% at Bob. Combining these features enables us to define a variation of the protocol, called SDI-HM-QKD, which exhibits semi-device-independent security.

%Consequently, as discussed in Section \ref{subsec:SDI}, we considered a variation of the protocol without post-selection to avoid side-channel attacks, showing how this still allows us to provide high key-rate at a metropolitan distance, even when Bob's device is considered completely untrusted.

\end{itemize}
 \noindent
 
Moreover, the security proof that we have established for a key distribution scheme based on the $\beta PM$ problem could also be applied to any one-way communication complexity problem with a boolean output. The results obtained in this article hence also pave the way to the study of other communication complexity problems with larger gaps between classical and quantum strategies, which would lead to even greater performances.
%
%We also point out that the Quantum Computational Time-lock framework and the proof technique developed for HM-QCT are based on a reduction to a communication complexity problem and that this reduction can be of independent interest in other quantum cryptographic contexts. 

\subsubsection{Technical contributions}
One of the main technical achievements has been to reduce Eve's general i.i.d.\ attack strategy, represented in Figure \ref{fig:Eve_general_strategy}, to a  strategy where she has no access to any quantum storage at the cost of an additive linear term in the noise parameter $\delta$, as formally proved in Theorem \ref{theo:boundP_guess}. Since this result is solely based on the fact that an eavesdropper has access to a noisy memory in Definition \ref{def:qmemory}, Theorem \ref{theo:boundP_guess} could be of independent interest for other protocols that exploit noisy storage.

Once we reduce to an eavesdropper with no quantum memory, a central result of our work is the exploitation of the communication gap between quantum and classical strategies to build a secure key distribution protocol.   In particular, the security reduction to the communication complexity of the  $\beta PM$ problem  cannot be done directly. First, since Alice is sending $m$ copies of the same $n$-dimensional quantum state, the amount of information that she is leaking to Eve about the input $x$  is at most $m \log(n)$ bits thanks to the Holevo bound. This simply reduces the security proof to the study of the \textit{information complexity} \cite{bravermanInformationEqualsAmortized2011} of the classical $\beta PM$ problem, a quantity which describes the amount information exchanged about the input needed to reliably solve the complexity problem. 
%However, we needed further analysis to relate the minimum amount of classical information transmitted to solve the $\beta PM$ problem to the lower bound in terms of classical communication  cost. In particular, we exploited a known one-way compression scheme, described in Appendix \ref{appendix:useful}, to  map the   communication complexity to the information complexity in the one-way setting. The  bound obtained is formally presented in Lemma \ref{lem:ICtoDC}.  
%Note that the result we derive is not specific to the $\beta PM$ problem and could be re-applied to other one-way communication complexity problems.

%While many articles on communication complexity focus on asymptotic scalings of input dimensions, our work necessitates precise knowledge of the lower bounds of communication complexity in the non-asymptotic regime. With these bounds, we can accurately determine how the guessing probability scales with the dimension and the number of copies sent by Alice, as detailed in Theorem \ref{theo:P_guess_fin}.

Through mapping communication complexity to information complexity in the one-way setting in Lemma \ref{lem:ICtoDC}, we demonstrate in Theorem \ref{theo:P_guess_fin} that Eve's one-round guessing probability is safely bounded away from $1$ when Alice sends $\mathcal{O}\left( \frac{\sqrt{n}}{\log(n)}\right)$ copies of the quantum state.

%Finally, to provide a comprehensive analysis, we also investigate in Section \ref{sec:upperbound} the achievable key rate obtained when considering the best-known protocol for the $\beta PM$ problem. Specifically, we demonstrate the potential to surpass the rate-loss limitations of standard QKD by implementing the HM-QCT protocol for experimentally feasible values of $n$.

% Finally, for the sake of clarity, we present in the following an informal summary of these asymptotic results.
%
%\begin{theorem}[HM-QCT protocol $\sim$ informal]
%Let $C>0$ and $0<k<1$. Suppose Alice and Bob execute the HM-QCT protocol sending  $m = C\left(\frac{\sqrt{n}}{\log(n)}\right)^{1-k}$ copies of an $n$-dimensional quantum state at each round of the protocol. Assuming that an eavesdropper Eve behaves independently and identically in each round of the protocol and that she has access to a $(t_{comp},\delta)$-noisy quantum memory with $\delta \ll 1$,  her information about the raw key approaches zero as $I(A:E)\sim\left(\frac{\log(n)}{\sqrt{n}}\right)^{k/3}$. Moreover, for any  value $T$ of the total transmittance of the quantum channel  the achievable key rate  can reach the classical capacity of one bit per channel use provided that Alice sends $m \gg \frac{1}{T}$ copies per channel use.
%\end{theorem}

%
%

\subsection{Previous work}
Communication complexity \cite{yaoComplexityQuestionsRelated1979} is a model of computation where two parties, Alice with input $x$ and Bob with input $y$, collaborate to compute with high probability the value of $f(x,y)$, where $f$ is a function (or relation) defining the computational problem that the players have to solve.
An exponential separation in the required amount of communication between quantum and classical strategies has been already shown experimentally \cite{kumarExperimentalDemonstrationQuantum2019} and then used to build a private quantum money scheme \cite{kumarPracticallyFeasibleRobust2019}. However,  to the best of our knowledge, this work describes the first explicit quantum key distribution protocol that guarantees security based on this exponential separation.

On the other hand, it is not the first time that  protocols have relied on  physical limitations of the quantum storage capabilities to extend the functionality of QKD. In the quantum bounded-storage model, for example, by limiting the amount of quantum information that an eavesdropper can store and process, QKD protocols can be designed to allow for higher error rates compared to the standard model with unbounded adversaries \cite{damgardTightHighOrderEntropic2007}.  An additional way to provide high resilience to noise, either caused by a malevolent Eve or simply environmental, is to perform QKD with high dimensional quantum states \cite{cozzolinoHighDimensionalQuantumCommunication2019} in the standard security model.
However,  both these frameworks are still highly susceptible to loss,  since their security is  limited to send only one photon per channel use.

Another example is also the theoretical framework of Quantum Data Locking (QDL) \cite{lupoQuantumLockedKeyDistribution2014}, where the security  of communication schemes is based on the even stricter assumption that quantum storage fully decoheres (i.e. $\delta =0$) after some finite time.  Existing work on QDL is  either restricted to single-photon encoding \cite{lupoQuantumLockedKeyDistribution2014,lumQuantumEnigmaMachine2016}, with  limitations in terms of loss-tolerance, or resorts to constructions based on random coding arguments \cite{lupoContinuousvariableQuantumEnigma2015} for which practical  decoding measurements with current technologies are not possible.

Security models with limitations  in the accuracy of the storage of  quantum states do not solely focus on key distribution schemes. The noisy-storage model \cite{wehnerCryptographyNoisyStorage2008a,koenigUnconditionalSecurityNoisy2012} is indeed a well-known security model, which generalizes the quantum bounded-storage model. It has been used to prove security of two-party protocols such as oblivious transfer \cite{santosQuantumObliviousTransfer2022} and bit commitment \cite{almeidaBriefReviewQuantum2014}, for which full unconditional security is impossible \cite{loWhyQuantumBit1998,buhrmanCompleteInsecurityQuantum2012}. Experimental demonstrations of these protocols were moreover performed, with typical hardware used in key distribution protocols, both for discrete \cite{ngExperimentalImplementationBit2012,ervenExperimentalImplementationOblivious2014} and continuous variable protocols \cite{furrerContinuousvariableProtocolOblivious2018}.
 However,  unlike the QCT model,  both QDL and the noisy-storage model do not rely on any computational assumptions, but they force the adversary to store the quantum states by intentionally delaying the classical post-processing. While this solution is enough to prove security, it has clear setbacks in the speed of the key exchange which is an important practical consideration.

\section{Preliminaries}
\subsection{General notation}
We reserve capital letters for random variables and distributions, calligraphic letters for sets, and lowercase letters for elements of sets. Let $\mathcal{S}$ be a set. We use $\Delta(\mathcal{S})$ to denote the family of all probability distributions on $\mathcal{S}$.
We use $\mathcal{D}(\H)$, $ \mathcal{L}(\H)$ and $\mathcal{P}(\H)$ to denote the sets
of density operators, square linear operators and positive semidefinite operators, respectively, acting on a  finite  dimensional Hilbert space $\H$.  Moreover, we  will use extensively the notation $d_{[\cdot]} \coloneqq \operatorname{dim} [\H_{[\cdot]}]$. 
The trace norm on $\L(\H)$ is defined as  $\lVert \sigma \rVert_1 \coloneqq \Tr\sqrt{\sigma \sigma^{\dag}}$.
% We call $\operatorname{Bin}(k, n, p) := \binom{n}{k} p^k (1-p)^{n-k} $ the Binomial distribution, where $\binom{n}{k} = \frac{n!}{k!(n-k)!}$.

Consider a classical random variable $A$ with distribution $P_A$ on some set $\mathcal{A}$. Since we are going to treat classical and quantum variables with the same formalism, it is useful to view $A$ as a particular case of a quantum system. We shall identify the classical values $a \in \mathcal{A}$ with some fixed orthonormal basis $\ket{a}$ on some Hilbert space $\H_{\mathcal{A}}$. The random variable A can then be identified with the quantum state $\rho_A = \sum_{a \in \mathcal{A}} P_A(a) \ket{a}\! \bra{a}.$
We can extend this representation to hybrid settings where the state $\rho_a$ of a quantum system $\H_Q$ depends on the value of $a$ of a classical random variable $A$. Such a state is called a \textit{classical-quantum state}, or simply \textit{cq-state}, and takes the form $\rho_{AQ}= \sum_{a\in \mathcal{A}} P_A(a)\ket{a}\! \bra{a} \otimes \rho_a$.

%An equivalent definition can be used to describe quantum operations, i.e. completely positive trace preserving (CPTP) maps,  whose outcomes are partly classical.
%A CPTP map $\mathcal{E}: \, \mathcal{D}\left(\H_1\right) \rightarrow\mathcal{D}\left(\H_\mathcal{A}\right) \otimes \mathcal{D}\left(\H_2\right)  $ is said to be classical on $\H_{\mathcal{A}}$ if it can be written as $\mathcal{E}(\sigma) = \sum_{a} \ket{a}\!\bra{a} \otimes \mathcal{E}^a(\sigma)$,
% where for any $a \in \mathcal{A}$, $\mathcal{E}^a$ is a trace non-increasing complete positive map from $\mathcal{D}\left(\H_1\right)$ to $\mathcal{D}\left(\H_2\right)$, with the additional condition that $\sum_a \mathcal{E}^a$  is trace-preserving.
%One should observe that a measurement on $\H_1$ with outcomes in $\mathcal{A}$ can be seen as a CPTP  map from $\mathcal{D}\left(\H_1\right)$ to $\mathcal{D}\left(\H_{\mathcal{A}}\right)$ classical in $\H_{\mathcal{A}}$. In general we refer to CPTP maps as quantum channels. 
\subsection{Classical and quantum information theory}
We need to define some notions of classical and quantum information theory. First, we  quantify the amount of information shared between two random variables $A$ and $B$ with distribution $P_{AB} \in \Delta (\mathcal{A} \times \mathcal{B})$ and  marginal distributions $P_A$ and $P_B$ respectively. We call  $I(A:B) \coloneqq H(A)-H(A|B)$, the \emph{mutual information}, where $H(A) \coloneqq -\sum_a P_A(a) \log (P_A(a))$ is the \emph{Shannon entropy} and
$ H(A|B) \coloneqq  -\sum_{a,b} P_{AB}(a,b) \log \left(\frac{P_{AB}(a,b)}{P_A(a)}\right)$  is the conditional entropy. Throughout this work, the function $\log$ will denote the logarithm base 2. 

Given any cq-state $\rho_{AQ}$, another  useful quantity  in quantum cryptography is the probability  of guessing the random variable $A$ for an adversary holding a quantum system $Q$, given by $P_{\text{guess}}(A|Q) \coloneqq \max_{\Pi} \sum_a P_A(a) \Tr[ \Pi(a) \rho_a]$,
 where we maximize over all POVMs $\Pi :A \rightarrow \mathcal{P}(\H_{Q})$.  Finally, we can define a conditional entropy, called the \emph{min-entropy}, given by $ H_{\min}(A|Q) \coloneqq - \log(P_{\text{guess}}(A|Q))$.

Now we introduce the generalization of Shannon entropy for quantum states, called the \emph{von Neumann entropy}. The von Neumann entropy of $\rho \in\D(\H_A)$ is $H(A)_{\rho}:=-\Tr[\rho\log(\rho)]$.
One can notice that by considering a classical state we recover back the Shannon entropy. 
For a bipartite state $\rho_{AE}\in\mathcal{D}(\H_A \otimes \H_E)$, we use the notation
$\rho_{E}$ for $\Tr_A[\rho_{AE}]$  and define the \emph{conditional von Neumann entropy} of system $A$ given system $E$ when the joint system is in the state $\rho_{AE}$ by $H(A|E)_{\rho}:=H(AE)_{\rho}-H(E)_{\rho}.$
We can finally define the \emph{quantum mutual information} as $I(A:B) \coloneqq H(A)_{\rho} - H(A|B)_{\rho}$.

\subsection{One-way Communication and Information Complexity}
\label{subsec:CC}
Communication complexity is a computation model introduced by Yao \cite{yaoComplexityQuestionsRelated1979}. It involves two players, Alice and Bob, who receive inputs: Alice receives $x$ from set $\mathcal{X}$ and Bob receives $y$ from set $\mathcal{Y}$. Their objective is to compute the value of $f(x,y)$ with high probability using allowed communication methods (classical or quantum). In this article, we focus on one-way settings, where only Alice can send messages to Bob. The message sent by Alice to Bob is called the \emph{transcript}, and Bob's final guess of $f(x,y)$ is called the \emph{output}.
In the public-coin model, they share a random string $r$, while in the private-coin model, they have private random strings $r_A$ and $r_B$. 
We start by defining the \emph{communication cost} of a protocol and the \emph{one-way distributional complexity} in the public-coin setting.

%
%\textbf{Notation}: Given a protocol $\pi$,  we call the transcript $\pi(x)$ the concatenation of the public randomness with  the message that Alice sends during the execution of $\pi$. When referring to the random variable denoting the transcript, rather than a specific transcript, we will use the notation $\Pi$. Lastly, we will denote Bob's final output of the protocol by $\pi(x,y)$.

\begin{definition}[Communication Cost]
\label{def:communicationcost}
The communication cost  of a public coin protocol $\pi$, denoted by $CC(\pi)$, is the maximum number of bits that can be transmitted in any run of the protocol.
\end{definition}
%\begin{definition}[Communication Cost]
%The communication cost  of a one-way public-coin protocol $\pi$, denoted by $CC(\pi)$, is the maximum number of bits that can be transmitted in any run of the protocol, i.e.
%
%\begin{equation}
%CC(\pi) \coloneqq \max_x \mathbb{E}_{R \leftarrow r} |\pi(x,r)|\;.
%\end{equation}
%
%An equivalent definition can be given for the communication cost of deterministic one-way protocols:
%\begin{equation}
%CC^D(\pi) \coloneqq \max_x |\pi(x)|\;.
%\end{equation}
%\end{definition}

 \begin{definition}[One-way distributional complexity]
 \label{def:distributional}
 For a function $f \, : \, \mathcal{X} \times \mathcal{Y} \rightarrow \mathcal{Z}$, a distribution $\mu \in \Delta(\mathcal{X} \times \mathcal{Y})$ and a parameter $\epsilon >0$, we define the one-way distributional complexity $D^1_{\mu}(f, \epsilon)$ as the communication cost of the cheapest one-way deterministic protocol for computing $f$ on inputs sampled according to $\mu$ with error $\epsilon$, i.e.
 
 \begin{equation}
 D^1_{\mu}(f, \epsilon) \coloneqq \min_{\pi: \,P_{(X,Y)}[\pi_{out}(x,y), \neq f(x,y)]  \le \epsilon} CC(\pi)\;,
 \label{eq:defin_distibutional}
 \end{equation} 
 where $\pi_{out}(x,y)$ describes Bob's output.
% \begin{equation}
% D^1_{\mu}(f, \epsilon) \coloneqq \min_{\pi: \,\bm{P}_{(x,y)\sim \mu}[\pi(x,y), \neq f(x,y)]  \le \epsilon} CC(\pi)\;.
% \end{equation}

 \end{definition}

 \noindent
We also consider different relevant quantities that apply an information-theoretic formalism to computational settings. 

 \begin{definition}[External Information Cost]
Fix a one-way private-coin communication protocol $\pi$ on inputs $\mathcal{X} \times \mathcal{Y}$ and a distribution $\mu \in \Delta(\mathcal{X} \times \mathcal{Y})$. The one-way external information cost of $\pi$ with respect to $\mu$, denoted by $IC^1_{\mu}(\pi)$ is defined as
\begin{gather}
IC^1_{\mu}(\pi) \coloneqq I(\Pi:X)\;,
\end{gather}
where $\Pi = \Pi(X,R_A)$ describes the transcript of the protocol.
 \end{definition}

\noindent Intuitively, the external information cost captures  how much information an external viewer who does not know the inputs learns about $X$. Similarly to the communication version, we can define the information complexity of a problem as the infimum over all possible protocols.

\begin{definition}[One-way external information complexity]
\label{definition:info_complex}
Let $\pi$  be a one-way private-coin protocol on inputs $\mathcal{X} \times \mathcal{Y}$ and  $\mu \in \Delta(\mathcal{X} \times \mathcal{Y})$. The one-way  external information complexity of $f$ with error tolerance $\epsilon$ is defined as the infimum of the one-way external information cost over all private-coin protocols $\pi$ for computing $f$  that achieve an error no larger than $\epsilon$  with respect to $\mu$:

\begin{gather}
\label{eq:info_complex}
IC^1_{\mu}(f, \epsilon) \coloneqq \inf_{\pi: \,P_{(X,Y),R_A,R_B}[\pi_{out}(x,y,r_A,r_B) \neq f(x,y)]  \le \epsilon} IC^1_{\mu}(\pi)\;,
\end{gather}
where $\pi_{out}(x,y,r_A,r_B)$ describes Bob's output.
\end{definition}
In this case we only  considered a private-coin model, since one can see that any public randomness can be simulated by a private-coin model: Alice can send to Bob a portion of $r_A$ together with the private-coin transcript. Now they can use this portion as shared randomness $r$. However, while this extra step increases the communication cost, it doesn't affect the external information cost.
\subsubsection{From distributional to information complexity}
As it turns out in Section \ref{Sec:security_analysis}, simply having a gap between quantum and classical communication complexity is insufficient
for cryptographic applications within the QCT model. Instead, what is crucial is the presence of a gap between quantum communication complexity and classical information complexity. 

 In \cite{harshaCommunicationComplexityCorrelation2010}, the authors demonstrated that it is possible to compress each message of a protocol to approximately its contribution to the external information cost plus some additional constant term.
While the authors focused only on the scaling laws, we carefully derived all the specific constants for the compression scheme.

\begin{restatable}[Mapping to information complexity]{lemma}{mappingtoinfo}
\label{lem:ICtoDC} 
Let $\epsilon, \, \delta_2>0$, $\mu \in \Delta(\mathcal{X} \times \mathcal{Y})$ and $f:\mathcal{X} \times \mathcal{Y} \rightarrow \mathcal{Z}$. Then
\begin{equation*}
IC^1_{\mu}(f,\epsilon) \ge \frac{\delta_2}{2} D^1_{\mu}(f,\epsilon + \delta_2) -6 \;.
\end{equation*}
\end{restatable}
\noindent
We refer to Appendix \ref{appendix:comp_scheme} for a description on how to derive this result from \cite{raoCommunicationComplexityApplications2020a}.

\subsection{\texorpdfstring{$\beta$}{TEXT}-Partial Matching problem}
\label{subsec:betapm}
In this subsection we shall present the quantum communication complexity problem that we want to use to build a key distribution protocol. Let $n \in \mathbb{N}$. We use the notation $[n] = \{1,..., n\}$. In the following $n$ will be assumed to be even. A \emph{matching} $M$ is a set of pairs $(a,b) \in [n]^2$,  such that no two pairs contain the same index, where, each index is called a \emph{vertex} and a pair of vertices is called an \emph{edge}.  For example if $n=4$ then the set of  edges $\{(1,2),(3,4)\}$ or $\{(2,3)\}$  are valid matchings whereas $\{(1,2),(2,3)\} $ are not. See Figure \ref{fig:betaPM} for a pictorial representation.

\begin{figure}[htb!]
\vspace{0.4cm}
\centering
         \includegraphics[width=0.6\textwidth]{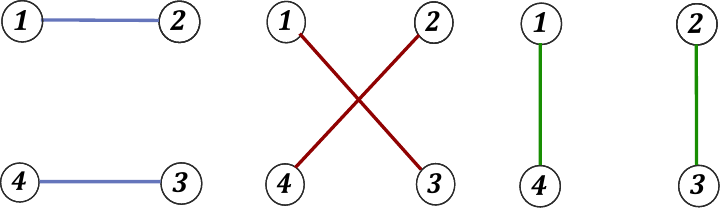}

     \caption{Illustration of a set of perfect matchings for size $n = 4$. For example, considering $x = 1001$, $\omega = 11$, for the first perfect matching in blue we have $Mx = \begin{bmatrix} x_1 \oplus x_2 =1\\ x_3 \oplus x_4=1 \end{bmatrix}$, resulting in $a=0$. }
    
        \label{fig:betaPM}
\end{figure}

We say $M$ is a \emph{$\beta$-matching} if in addition $|M| = \beta n$. 
The $\beta PM$ problem is built around a $\beta$-matching $M$, that constitutes part of the input given to Bob. $M$ consists of a sequence of $\beta n$ disjoint edges $(i_1, j_1) ... (i_{\beta n}, j_{\beta n})$ over $[n]$.  
 We will call  $\mathcal{M}_{\beta n}$ the set of all $\beta$-matchings on $n$ bits:  if $\beta = \frac{1}{2}$  the matching is called \textit{perfect} and if $\beta < \frac{1}{2}$  the matching is called \textit{partial}. 
$M$ can be represented as a $ \beta n \times n$  matrix with  only a single one in each column and two ones per row, namely at position $i_l$ and $j_l$  for the $l$-th row of matrix $M$. Let $x \in \{0, 1\}^n$, applying the matching $M$ to $x$ leads to the $\beta n$-bit string  $v = v_1,..., v_l,..., v_{\beta n}$  where $v_l = x_{i_l} \oplus x_{j_l}$. Finally, we call $a^{\beta n}$ a vector of dimension $\beta n$ with value $a$ in each component.
Using the notation above, we can finally define the $\beta PM$ problem:\\[0.2cm]
\textbf{Alice's input}: $x \in \{0,1\}^n$. \\
\textbf{Bob's input}: $M \in\mathcal{M}_{\beta n}$ and $\omega \in \{0,1\}^{\beta n}$. \\
\textbf{Promise P}: given a bit  $a \in  \{0,1\}$,  then  $\omega = M x \oplus a ^{\beta n}$.\\
\textbf{Communication Model}: Classical or Quantum one-way communication between Alice and Bob.\\
\textbf{Goal}: Bob outputs $b=a$ with high probability.\\[0.2cm]
\noindent
We shall call for clarity $\mathcal{X} \coloneqq  \{0,1\}^n $,  $y \coloneqq (M, \omega) $ and $\mathcal{Y} \coloneqq \mathcal{M}_{\beta n} \times\{0,1\}^{\beta n} $.
Moreover, we define the (partial) function $\beta PM: \mathcal{X} \times \mathcal{Y} \rightarrow \{0,1\}$ as  the function that randomly picks an element from the vector $M x \oplus \omega$.\\[0.2cm]
\textbf{Input distribution:} we call $\mu  \in \Delta(\mathcal{X} \times \mathcal{Y})$ the  input probability distribution uniform over $x\in \{0,1\}^n$ and $M\in\mathcal{M}_{\beta n}$. The inputs $x$ and $M$ together determine the $\beta n$-bit string $v = Mx$. To complete the input distribution, with probability $1/2$ we set $\omega = v$ and with probability $1/2$ we set $\omega = \bar{v}$.

Finally,  one can derive from \cite{gavinskyExponentialSeparationsOneway2007}  the prefactors of the scaling law for the one-way distributional complexity of the $\beta PM$ protocol.
 
\begin{theorem}
\label{theo:exponential_sep}
Let $\beta \in (0, 1/4]$, $\forall \epsilon \in (0,\frac{1}{2}]$. Then 
 \begin{equation}
 D_{\mu}^1(\beta PM,\epsilon)\ge k(\epsilon) \sqrt{n} + d(\epsilon)\;,
 \end{equation}
where 
\begin{equation}
\label{eq:k_and_d}
k(\epsilon)= \frac{4\gamma}{25 \sqrt{\beta}} \left(\frac{1}{2} - \epsilon \right)^2 \hspace{0.3cm} \text{and}\hspace{0.3cm} d(\epsilon)= 2\log\left( \frac{1}{2} - \epsilon \right) + 2(\log(2) - \log(5)) \;,
\end{equation} 
with $\gamma = \frac{1}{8e}$.
\end{theorem}
\begin{proof}
A complete description of how to derive this theorem from \cite{gavinskyExponentialSeparationsOneway2007} is given in Appendix \ref{appendix:deriv_th}.
\end{proof}

\section{Key Establishment Protocol}
\label{sec:key_estab}
\subsection{Security model and definitions}

Considering the novelty of our hybrid security model, the assumptions on the resources of an adversary and the  security properties that can be achieved in this model must be described thoroughly.
%\subsubsection{QCT-model}

In the QCT construction, authorized parties, Alice and Bob, are assumed to be connected via a noiseless and authenticated classical channel and an insecure quantum channel. An adversary, Eve, is assumed to have full access to the input of Alice and Bob's communication channels. Every classical (quantum) message communicated between Alice and Bob over the classical (quantum) channel can be wiretapped by Eve and stored in classical (quantum) memory. With this
pessimistic setting for Eve's channel, we are in a similar set-up as strong data locking \cite{lupoQuantumLockedKeyDistribution2014,lupoQuantumDataLocking2015a} wherein an adversary Eve receives direct inputs from Alice. As stated in Section \ref{sec:intro_QCT}, the QCT model is based on two main assumptions on the power of an eavesdropper: a computational assumption (see Definition \ref{def:encr_scheme}) and a noisy-storage assumption (see Definition \ref{def:qmemory}). 

We start by stating the type of computational assumption needed to prove security in our scheme. What Alice and Bob need is a \emph{semantically secure}  symmetric encryption scheme  against adaptive chosen-ciphertext attacks (CCA2) for a time at least $t_{comp}$. Semantic security means that it is computationally unfeasible for an eavesdropper to learn any partial information about a plaintext from the corresponding ciphertext (see \cite{katzIntroductionModernCryptography2014} for a formal definition.)
This implies  that  the encrypted message $\texttt{Enc}_k(m)$ is (computationally) indistinguishable from a completely random string until at least a time $t_{comp}$. Furthermore, the security  against adaptive chosen-ciphertext attacks ensures another required property: \emph{non-malleability} \cite{katzIntroductionModernCryptography2014}. In simple terms, an encryption scheme is called non-malleable if one cannot feasibly manipulate a given ciphertext in such a way that it produces another ciphertext, which, when decrypted, yields a plaintext related to the original.
Finally, the desired security for the hybrid key distribution protocol is based on the \emph{trace distance criterion} \cite{rennerSecurityQuantumKey2008}, a standard criterion to prove information-theoretic security for  quantum key distribution.

%\begin{definition}[$\epsilon$-secure key]
%%A $\epsilon)
%%$-secure protocol is a quantum key distribution protocol between Alice and Bob where they pre-share a secret $S$ using a $t_{comp}$-secure encryption scheme.  
%Let us consider  the $ccq$-state $\rho_{K_AK_BE} $ as the output of the  quantum key distribution protocol, where $E$ is the quantum system of  the adversary, $K_A$ and $K_B$ are Alice's and Bob's final keys respectively and $\mathcal{K}$ is the set  of possible keys. An  $\epsilon$-secure key  must then satisfy the following properties:\\[-0.4cm]
%\begin{itemize}
%\item \textbf{$\epsilon_{cor}$-correct}: \,\,\,$Pr(K_A \neq K_B) \le \epsilon_{cor}$
%\item \textbf{$\epsilon_{sec}$-secrect}:\,\,\, $\frac{1-p^{\perp}}{2} \lVert \rho_{K_AE} - \frac{1}{|\mathcal{K}|} \sum_{k \in \mathcal{K}} \ket{k}\!\bra{k} \otimes \rho_E \rVert_1 \le \epsilon_{sec}$,
%\end{itemize}
% where $p^{\perp}$ is the probability that the protocol aborts and the final security parameter is simply $\epsilon = \epsilon_{cor}+ \epsilon_{sec}$.
%\end{definition}
%\noindent
%These two conditions guarantee that the probability of Alice and Bob not sharing the same key at the end of the protocol is extremely low, while at the same time a malicious eavesdropper cannot distinguish the final key from a completely random string.
%

\subsection{Protocol description}
\label{sec:protocol_desc}
Now that we have introduced all the crucial ingredients, we can
present and analyze our protocol.
The main building block for our construction is an explicit quantum communication protocol that solves the $\beta PM$ problem by simply sending a constant number of $n$-dimensional quantum states \cite{gavinskyExponentialSeparationsOneway2007}.

\subsubsection{\texorpdfstring{$\beta PM$}{TEXT} quantum protocol}
Alice sends a uniform superposition of her bits to Bob:
 \begin{equation}
 \label{eq:psi_x}
 \ket{\psi_x} = \frac{1}{\sqrt{n}} \sum_{i=1}^n (-1)^{x_i} \ket{i}\;.
 \end{equation}
  Bob completes his $\beta n$ edges to a perfect matching in an arbitrary way and measures with the corresponding set of $n/2$ rank $2$ projectors, where for an edge $(a,b)$ the projector is $P=\ketbra{a} + \ketbra{b}$. With probability $2\beta$ he will receive an output corresponding to one of the edges $(i_l, j_l)$ from his input $\beta$-matching $M$. The state then collapses to $\frac{1}{\sqrt{2}} ((-1)^{x_{i_{l}}} \ket{i_l} + (-1)^{x_{j_l}} \ket{j_l},$
 from which Bob can obtain the bit $v_l = x_{i_l} \oplus x_{j_l}$  using a measurement containing projectors $\{\ket{+}\!\bra{+}, \ket{-}\!\bra{-}\}$, where $\ket{+} = (\ket{i_l} + \ket{j_l})/\sqrt{2}$ and $\ket{-} = (\ket{i_l}- \ket{j_l})/\sqrt{2}$, and immediately retrieve the bit $a$. With probability $1-2\beta$,  instead, he will receive an output that doesn't correspond to any edge of the $\beta$-matching $M$: in this case, he immediately outputs $b=\perp$, aborting the protocol.
 %By repeating a constant number of times they can achieve correctness $\-\epsilon_{cor}$ for any small constant $\epsilon_{cor}$.
 One important point is that Bob can perform his measurement with only two single photon detectors, since he can pre-route, in accordance with $(M; \omega)$ the output of the $n$ beamsplitters. See Figure  \ref{fig:implementation} for a pictorial representation.

In practice the quantum channel  and detectors will be subject to loss and errors.  What Alice and Bob can implement is a practical version of the $\beta PM$ protocol, described in detail in Appendix \ref{appendix:practical_prot}, where they compensate for the loss by sending several copies of the same state $\ket{\psi_x}$.

 \begin{figure}[htb!]
\centering
\includegraphics[width=0.75\textwidth]{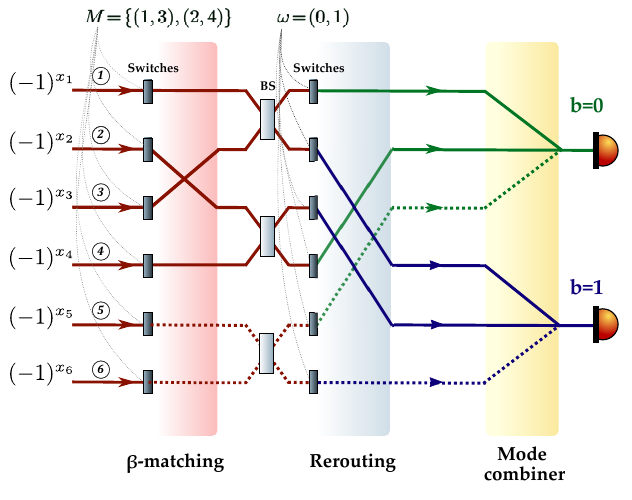}  
     \caption{Illustration of a possible implementation of Bob's decoding with $n=6$ spatial modes and $\beta = 1/3$.  In the $\beta$-matching part Bob uses his knowledge of $M$ to control each switch and direct each mode to the corresponding  beam splitter (BS). The modes with dotted lines  are blocked instead, since they don't correspond to any vertex of the partial matching. Then, in the rerouting part, he  reorders the modes based on $\omega$. Finally, thanks to a mode combiner, he directs  the first (second) half of the modes to the first (second) detector.}
\label{fig:implementation}
\end{figure}

\subsubsection{HM-QCT key distribution scheme}
 Now that we have described the main building block, we are ready to present our hybrid key distribution protocol.\\
\rule{\textwidth}{0.2mm}\\[0.1cm]
HM-QCT Protocol\\[-0.1cm]
\rule{\textwidth}{0.2mm}
\\[-0.7cm]
\begin{enumerate}
\setlength\itemsep{-0.4em}
\item[] \textbf{Parameters:}
\\[-0.65cm]
\begin{itemize}
\setlength\itemsep{-0.2em}
\item[-] dimension $n$ of the problem
%\item[-] $t$ :time-lock of the encryption scheme
%\item[-] $\delta$: parameter of Eve's quantum memory
\item[-] number of copies $m$
\item [-] number of rounds $l$.\\[-0.3cm]
\end{itemize}
\item \textbf{Data Generation:} 
\\[-0.65cm]
\begin{itemize}
\setlength\itemsep{-0.3em}
\item[-]  Alice generates and stores $\vec{x}= \left(x_1,...x_{l}\right)$ and $\vec{y}= \left(y_1,...y_{l}\right)$  from the probability  distribution $\mu^l \in \Delta^l(\mathcal{X} \times \mathcal{Y})$. She then computes and stores the string $\vec{a} = (a_1,... ,a_l)$, where $a_j= \beta PM(x_j,y_j)$.\\[-0.3cm]
\end{itemize} 

\item \textbf{QCT exchange }
\\[-0.65cm]
\begin{itemize}
\setlength\itemsep{-0.2em}
\item[-] Alice and Bob run $\texttt{Gen}$ and obtain a shared secret $k$.
\item[-] Alice sends $\texttt{Enc}_k(\vec{y})$ to Bob.
\item[-] Bob decrypts $\texttt{Enc}_k(\vec{y})$ using $\texttt{Dec}_k$,  obtaining $\vec{y}$.\\[-0.3cm]

\end{itemize} 
\item \textbf{Quantum communication} 
\\[-0.65cm]
\begin{itemize}
\setlength\itemsep{-0.3em}
\item for $i=1; \, i\le l; \, i++$ 
\begin{itemize}
\setlength\itemsep{-0.3em}
\item Alice and Bob run the $\beta PM$ quantum protocol, with input $x_i$ and $y_i$. Bob stores the output $b_i$.
\end{itemize} 
\end{itemize} 
%\item \textbf{Sifting:}
%\begin{itemize}
%\item[-] Alice and Bob discard all rounds with ${b}_i =\perp$. Let $l' = F \cdot l$ be the number of remaining rounds (re-indexed as $1,...,l'$ ).
%\end{itemize}
\item \textbf{Sifting:}
\\[-0.65cm]
\begin{itemize}
\item[-] Alice and Bob discard all rounds with ${b}_i =\perp$.
%Let $l' = \freq \cdot l$ be the number of remaining rounds (re-indexed as $1,...,l'$), where $\freq$  is the probability  that a round of the protocol is conclusive.
\\[-0.3cm]
\end{itemize}

\item \textbf{Classical post processing:}
\\[-0.65cm]
\begin{itemize}
\item[-] Parameter estimation: Alice and Bob estimate
the quantum bit error rate (QBER) i.e.\ the error rate of a conclusive round, by revealing a part
of their string. 
%They abort the protocol if the QBER exceeds some threshold error probability $\epsilon_{th}$ that depends on the number $m$  of copies sent.
\item[-] Alice and Bob perform \textit{error correction} \cite{brassardSecretKeyReconciliationPublic1994} followed by \textit{privacy amplification} \cite{bennettGeneralizedPrivacyAmplification1995}  to distill
a secret key.\\[-0.8cm]
\end{itemize}

\end{enumerate}
\rule{\textwidth}{0.2mm}
%\\[0.2cm]
%\noindent
\noindent
\begin{remark}
The correctness of our protocol is ensured by the correctness of the $\beta PM$ protocol together with an extra step of error correction to deal with noise and loss present in practical scenarios. 
\end{remark}

\subsection{Security Analysis}
\label{Sec:security_analysis}
\subsubsection{Achievable key rate in the i.i.d. setting}

We now focus on how to derive an achievable key rate within our model. In this article we shall analyze the security of our key distribution protocol in the i.i.d.\ setting, i.e.\ a restricted case where the adversary Eve performs the same strategy independently on every round.  
In this setting, we can consider, without loss of generality, the most general attack from Eve on a single round of the protocol. 
It consists of immediately applying an encoding operation $\mathcal{E}: \L\left((\mathbb{C}^n)^{\otimes m}\right)  \rightarrow \L\left(\H_{\mathcal{Z}} \otimes \H_{Q_{in}} \right)$ statistically independent of $y$ due to the semantic security of the encryption scheme,  before storing the quantum state on her $(t_{comp},\delta)$-noisy quantum memory $ \Phi_{t_{comp}} :\, \L(\H_{Q_{in}}) \rightarrow \, \L(\H_{Q_{out}})$, following a similar strategy of \cite{furrerContinuousvariableProtocolOblivious2018}. Moreover, the non-malleability of the classical encryption scheme prevents Eve from running any homomorphic strategy, i.e.\ a quantum operation depending also on $\texttt{Enc}_k(y)$, which could eventually leak sensitive information. The encoding $\mathcal{E}$ also includes a classical outcome $Z$ that can, for instance, result from measuring part of the copies. 
Moreover, we consider that after the time $t_{comp}$, Eve is given the encrypted secret $y$, i.e.\ that  $\texttt{Enc}$ can be fully decrypted after $t_{comp}$, which is the most favorable case for Eve.

\begin{figure}[htb!]
\vspace{0.2cm}
\centering
\includegraphics[width= 0.85\textwidth]{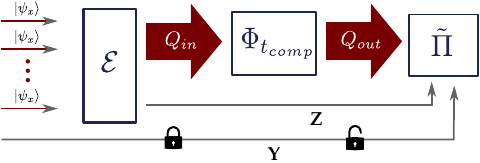}
\caption{General form of an attack of Eve. It consists in an encoding $\mathcal{E}$ that maps (conditioned on some classical outcome $Z$) the $m$ copies of  Alice's quantum state to the memory input $Q_{in}$. At  time $t_{comp}$, when she unlocks the secret $Y$, she decodes the key bit by performing the measurement $\tilde{\Pi}$ on $Q_{out}$ using both the secret $Y$ and  the classical outcome $Z$.}
\label{fig:Eve_general_strategy}
\end{figure}
One should note that this general strategy also includes the extreme strategies where Eve either simply stores the quantum input\footnote{Storing all copies simultaneously and measuring once the encrypted message $y$ is unlocked can be viewed a generalized version of the photon number splitting attack (PNS) \cite{lutkenhausQuantumKeyDistribution2002}. In PNS, eavesdroppers store extra photons in their quantum memory until they obtain the basis information, enabling them to execute the appropriate measurement.} $\ket{\psi_x}^{\otimes m}$, since we have never given any bound on the dimension of our memory, or the case where she measures all the copies immediately. Moreover, any strategy that consists of performing any general measurement at times different from $0$ and $t_{comp}$, even if surely suboptimal, can be described by this general strategy.
As a consequence of this setting, at the end of each round Alice and Bob have access to a realization of correlated classical random
variables $A$ and $B$, respectively, whereas the adversary Eve holds the quantum system $E=YZQ_{out}$.
The final joint state  for each round between Alice and Eve will therefore have the form

\begin{align}
\rho_{AYZ\Phi_{t_{comp}}(Q_{in})} = & \sum_{x,y,a} \mu(x,y)\delta_{\beta PM(x,y),a} \nonumber \\
& \ket{a}\!\bra{a} \otimes \ket{y}\!\bra{y} \otimes (\mathds{1}_{d_\mathcal{Z}} \otimes \Phi_{t_{comp}} )(\mathcal{E}(\rho_x))\;. 
\label{eq:joint_state}
\end{align}
At this point Eve performs the POVM  $\tilde{\Pi}: \{0,1\}\rightarrow \mathcal{P}(\H_{\mathcal{Y}}\otimes \H_{\mathcal{Z}}\otimes \H_{Q_{out}})$ on the output of the quantum memory to guess the bit $a$, making use of $y$  and the classical string $z$. 
Finally, we can lower bound the achievable key rate  under this general i.i.d.\ attack, depicted on Figure \ref{fig:Eve_general_strategy}.
Since the min-entropy lower-bounds the von Neumann entropy, we can lower bound the Devetak-Winter bound  \cite{devetakDevetakWinterDistillation2003}  and obtain the following achievable key rate
\begin{equation}
\label{eq:achievable_keyrate}
R^{\text{\tiny{QCT}}}_{\infty}\ge (1-\abort)\left(H_{min}(A|E) - h(\text{QBER})\right)\;,
\end{equation} 
where  $H_{min}(A|E)=-\log(\Pgen)$, $h$ is the binary Shannon entropy, and $\abort$  is the probability  that a round of the protocol is inconclusive. 
% One can notice that since in the QCT model Eve does not act on Alice's and Bob's quantum channel both $\abort$ and the $\text{QBER}$ are not affected by Eve.
In Appendix \ref{appendix:practical_prot} we have evaluated $\abort$ and $\text{QBER}$ as a function of the number of copies sent $m$ in a practical scenario, considering a  photonic implementation with only two detectors, as well as its variant that operates without post-selection (where $\abort = 0$).

\subsubsection{Bounding \texorpdfstring{$\Pgen$}{TEXT}}
We now  compute a lower bound for the achievable key rate in Eq.\ \eqref{eq:achievable_keyrate}. In particular  here we focus on computing $H_{min}(A|E)$.
We evaluate Eve's guessing probability $\Pgen$ in two steps. First we want to bound it with respect to a restricted strategy where she never uses a noisy quantum memory, but she performs immediately a joint measurement on the $m$ copies. We call such a strategy an \emph{immediate measurement strategy}.
The second step consists instead in deriving a bound on the guessing probability of this restricted strategy by exploiting the communication complexity gap between quantum and classical strategies for the $\beta PM$ protocol.
\subsubsection{Reduction to immediate measurement}
In this restricted scenario, following a standard post-measurement information strategy \cite{gopalUsingPostmeasurementInformation2010}, Eve performs an immediate measurement $\mathcal{Z} : $ $\L\left((\mathbb{C}^n)^{\otimes m}\right)  \rightarrow  \L\left(\H_{\mathcal{Z}}\right)$  on the input state $ \rho_x \coloneqq(\ket{\psi_x}\!\bra{\psi_x})^{\otimes m} $ and obtains a classical outcome  $z$.  At time $t_{comp}$, she unlocks $y$ and extracts the final guess by performing a classical decoding $\tilde{\Pi}_1$, that can be expressed as a POVM $\tilde{\Pi}_1: \{0,1\} \rightarrow \mathcal{P}( \H_{\mathcal{Y}} \otimes \H_{\mathcal{Z}})$.
The guessing probability can therefore be written as
\begin{equation}
 \Pim \coloneqq \max_{\tilde{\Pi}_1}\sum_{x,y,a} \mu(x,y)\delta_{\beta PM(x,y),a}\Tr [ \tilde{\Pi}_1(a)(\ketbra{y} \otimes \mathcal{Z}(\rho_x))]\;.
 \label{eq:P_guess_imm}
\end{equation}

 To  show the security reduction we first  prove the following useful (and more general) theorem.

\begin{theorem}
\label{theo:boundP_guess}
If $\lVert \Phi -  \mathcal{F} \rVert_{\diamond} \le \delta $, with  $\mathcal{F}$ being the completely mixing channel, then for any cqq-state $\rho_{AXQ}$ we have
\begin{equation}
P_{guess}(A|X\Phi(Q)) \le P_{guess} (A|X) + \delta\;.
\end{equation}
\end{theorem}
\begin{proof}
We can bound the guessing probability as follows
\begin{align*}
P_{guess}(A|X\Phi(Q))=& \max_{\Pi} \sum_a p(a) \Tr[\Pi(a)  \Phi (\rho_{XQ})] \\
%=& \max_{\Pi} \sum_a p(a) \left( \Tr[\Pi(a) (\Phi \!-\!\mathcal{F})(\rho_{XQ})] + \Tr[\Pi(a) \mathcal{F}(\rho_{XQ})]\right)\\
\le & \max_{\Pi} \sum_a \! p(a) \left(\lVert \Pi(a)\rVert_{\infty} \!\lVert (\Phi \!-\!\mathcal{F})(\rho_{XQ})\rVert_1 \!
+ \!\Tr[\Pi(a) \!\mathcal{F}(\rho_{XQ})]\!\right)\\
\le& \delta +  \max_{\Pi} \sum_a p(a)\Tr[\Pi(a) \mathcal{F}(\rho_{XQ})]\;,
\end{align*}
where we used the notation $\mathcal{N} (\rho_{XQ}) \coloneqq (\mathds{1}_{d_{\mathcal{X}}}\otimes \mathcal{N}) (\rho_{XQ})$ for  any quantum channel $\mathcal{N}$ acting only on $Q$. In the second line we used the Hölder's inequality,  while the last inequality is obtained by noticing that $\lVert M \rVert_{\infty} \le 1$ for any element of a POVM, the fact that $\sum_a p(a) = 1$, and the fact that  $\lVert (\Phi-\mathcal{F})(\rho_{AX})\rVert_1 \le\delta $, since we have  $\lVert \Phi -  \mathcal{F} \rVert_{\diamond} \le \delta $.
Finally, since $\mathcal{F}$ destroys all the quantum information in the system $Q$, we directly have $ \max_{\Pi} \sum_a p(a)\Tr[\Pi(a) \mathcal{F}(\rho_{XQ})] = P_{guess} (A|X)$ which concludes the proof. 
\end{proof}
Now from Theorem \ref{theo:boundP_guess} we simply have that for any encoding attack
\begin{align}
\Pgen &\le P_{guess}(A|YZ)+ \delta \nonumber \\
&\le \max_{\mathcal{Z}} \Pim + \delta \;, 
\end{align}
where we maximized over all possible Eve's immediate measurements $\mathcal{Z}$.
Hence, considering  $\delta <\! \! <1$, we have successfully reduced any general attack strategy to an immediate joint measurement on the $m$ multiple copies.

\subsubsection{Exploiting the complexity gap}
To finally estimate an upper bound to Eve's guessing probability we still have to study this restricted scenario.  Our approach for a full proof follows the idea that extracting a bit of the key with an immediate measurement strategy is as hard as solving the classical  $\beta PM$ problem.  In particular, Eve cannot do better than what one would get for the $\beta PM$ problem by sending $m \log(n)$ bits of information about the input $x$, where $m \log(n)$ bits is the maximum classical information one can extract from $m$ copies of  a $n$-dimensional quantum state thanks to the Holevo bound.

\begin{lemma}
\label{lemma:P_guess_imm}
$\forall \epsilon \in \left(0,\frac{1}{2}\right)$ if an immediate measurement strategy with $\Pim\ge 1-\epsilon$ exists, then Alice has sent $m$ copies of the quantum state  \eqref{eq:psi_x}, with
\begin{equation*}
m \ge \frac{IC_{\mu}^1(\beta PM, \epsilon)}{\lceil\log(n)\rceil}\;.
\end{equation*}

 \end{lemma}
\begin{proof}

Let's suppose there exists an immediate measurement strategy with \\$\Pim $ at most $1-\epsilon$, then we can transform this strategy into a classical protocol  to solve the  $\beta PM$ problem.
The transformation is straightforward, Alice generates  $m$ copies of the quantum state  \eqref{eq:psi_x}, then she immediately performs the measurement $\mathcal{Z}$ and sends the classical output $z$ to Bob who, after performing the final POVM $\tilde{\Pi}_1$  on $z$ and $y$, will output  the correct answer with probability at least $1-\epsilon$.   Note that  the string $z$ is the transcript of the protocol.
Since from Holevo's bound we know that $I(X; Z) \le m\lceil\log(n)\rceil$, by definition of $IC_\mu^1(\beta PM, \epsilon)$ we have
\begin{equation*}
m \ge \frac{IC_{\mu}^1(\beta PM, \epsilon)}{\lceil\log(n)\rceil}
\end{equation*}
that concludes the proof. 
\end{proof}

Finally,  thanks to the complexity gap between classical and quantum strategies, Theorem \ref{theo:P_guess_fin} ensures that Eve's guessing probability is safely bounded far from $1$  as  long as Alice is sending  $\mathcal{O}\left( \frac{\sqrt{n}}{\log(n)}\right)$ copies of the quantum state.

\begin{theorem}
 \label{theo:P_guess_fin}
Let us suppose $n \ge 4$. 
For any encoding attack Eve's guessing probability is bounded by 
\begin{equation}
\Pgen \le  \frac{1}{2} +2\left(\sqrt[3]{-q} + \sqrt{\frac{p}{3}} \right)+ \delta ,
\end{equation}
with
\begin{align*}
q & = 
\frac {-100}{\sqrt{n}} e \sqrt{\beta} ((m+1 )\lceil\log(n)\rceil +\frac{1}{\ln(2)}+ 6)\\
p & = \frac {-100}{\sqrt{n}} e\sqrt{\beta}\left(\log\left(\frac{5}{2}\right)- \frac{1}{\ln(2)}\right)\;.
\end{align*}
\end{theorem}
\begin{proof}
We first prove an useful lemma
\begin{lemma}
 \label{lemma:P_guess_fin}
$\forall \epsilon \in \left(0,\frac{1}{2}\right)$, $\forall \delta_2 \in 
\left( 0, \frac{1}{2} - \epsilon\right)$ if an encoding attack with \\$\Pgen \ge 1-\epsilon +
\delta$ exists, then Alice must have sent $m$ copies of the quantum state \eqref{eq:psi_x}, with 

\begin{equation}
\label{eq:copies-bound}
 m \ge \frac{\frac{\delta_2}{2}\left( \frac{1}{50e\sqrt{\beta}} \left(\frac{1}{2} - \epsilon
 -\delta_2 \right)^2 \sqrt{n} +2\log\left( \frac{1}{2} - \epsilon - \delta_2\right) -2\log\left(\frac{5}{2}\right)\right) - 6}{
\lceil\log(n)\rceil} \; .
\end{equation}
\end{lemma}

\begin{proof}
Let $\epsilon \in \left(0,\frac{1}{2}\right)$, $\delta_2 \in 
\left( 0, \frac{1}{2} - \epsilon\right)$. Let us suppose there exists an encoding attack with 
$\Pgen\ge 1-\epsilon +
\delta$. First, by using Theorem \ref{theo:boundP_guess}, we deduce
$ \max_{\mathcal{Z}} \Pim\ge 1 - \epsilon$.
Then we use Lemma \ref{lemma:P_guess_imm} to deduce
$m \ge \frac{IC_{\mu}^1(\beta PM, \epsilon)}{\lceil\log(n)\rceil}.$
Furthermore, from Lemma \ref{lem:ICtoDC} we obtain
$m \ge \frac{\frac{\delta_2}{2} D^1_{\mu}(f,\epsilon + \delta_2) - 6}{
\lceil\log(n)\rceil}.$
Finally,  we conclude the proof by showing that from Theorem \ref{theo:exponential_sep} we have

\begin{equation*}
 m  \ge \frac{\frac{\delta_2}{2}\left( k(\epsilon + \delta_2) \sqrt{n} + d(\epsilon + \delta_2)\right)- 6}{
\lceil\log(n)\rceil}\;,
\end{equation*}
with $k$ and $d$ defined in \eqref{eq:k_and_d}.

\end{proof}

\noindent
Now we are ready to prove Theorem \ref{theo:P_guess_fin}.  Let $x=\frac{1}{2} - \epsilon$ and $\delta_2 = \frac{x}{2}$. We can rewrite Eq.\ \eqref{eq:copies-bound} as
\begin{align}
m & \ge \frac{\frac{x}{4}\left( \frac{1}{50e\sqrt{\beta}} \left(\frac{x}{2} \right)^2 \sqrt{n} +2\log\left( \frac{x}{2}\right)-2\log\left(\frac{5}{2}\right)\right)- 6}{
\lceil\log(n)\rceil} \nonumber\\
& \ge \frac{\frac{x}{4}\left( \frac{1}{50e\sqrt{\beta}} \left(\frac{x}{2} \right)^2 \sqrt{n} +\frac{2}{\ln(2)}\left(1- \frac{2}{x}\right)-2\log\left(\frac{5}{2}\right)\right)- 6}{
\lceil\log(n)\rceil} \nonumber\\
& \ge \frac{\frac{1}{100e\sqrt{\beta}} \left(\frac{x}{2} \right)^3 \sqrt{n} +\frac{x}{2}\left(\frac{1}{\ln(2)}-\log\left(\frac{5}{2}\right)\right) -\frac{1}{\ln(2)} - 6}{
\lceil\log(n)\rceil}
\end{align}
where in the second inequality we used the fact that $\ln(t)\ge 1 - \frac{1}{t}$ for $t>0$.
By contraposition, Lemma \ref{lemma:P_guess_fin} implies that for any encoding attack acting on $m$ copies, with 
\begin{equation}
\label{eq:xeqzero1}
m = \frac{\frac{1}{100e\sqrt{\beta}} \left(\frac{x}{2} \right)^3 \sqrt{n} +\frac{x}{2}\left(\frac{1}{\ln(2)}-\log\left(\frac{5}{2}\right)\right) -\frac{1}{\ln(2)} - 6}{
\lceil\log(n)\rceil} -1 \;,
\end{equation}
Eve's guessing probability is bounded by
\begin{equation}
\label{eq:Pgen2}
    \Pgen < \frac{1}{2}+x +
\delta\;.
\end{equation}
We now have to find the real zero of Eq.\ \eqref{eq:xeqzero1} by using Cardan's method.
We first  rewrite  Eq.\ \eqref{eq:xeqzero1} in in the canonical form 
\begin{equation}
z^3 + pz +q=0\;,
\label{eq:xeqzero}
\end{equation}
where 
\begin{align*}
z & =\frac{x}{2}\;,\hspace{0.3cm}
q = 
\frac{-100}{\sqrt{n}}e\sqrt{\beta} 
\left((m+1 )\lceil\log(n)\rceil + \frac{1}{\ln(2)}+6\right)\;,\\
p & = \frac{-100}{\sqrt{n}}e\sqrt{\beta}\left(\log\left(\frac{5}{2}\right)- \frac{1}{\ln(2)}\right)\;.
\end{align*}
\noindent
We now observe that $q<0$ and, since $\left(\log\left(\frac{5}{2}\right) - \frac{1}{\ln(2)}\right) <0,$ it follows that
$p>0$. Consequently, this implies that $\Delta \coloneqq -( 4p^3 + 27q^2)$ is negative.
Therefore, thanks to Cardan's method, the zero of Eq.\ \eqref{eq:xeqzero}  expressed in the variable $x$ is
 \begin{align}
x& = 2^{1- \frac{1}{3}}\left(\sqrt[3]{-q + \sqrt{\frac{-\Delta}{27}}} + 
\sqrt[3]{-q - \sqrt{\frac{-\Delta}{27}}}\right)
\nonumber \\
& \le 2^{1- \frac{1}{3}}\left(\sqrt[3]{-q + \sqrt{\frac{-\Delta}{27}}}\right) \nonumber \\
& \le 2 \left( \sqrt[3]{-q} + \sqrt{\frac{p}{3}} \right)\;.
\label{eq:zero_x}
 \end{align}
Where in the second line we used the fact that the second term with $\sqrt[3]{\cdot}$ was negative. The last line was obtained using the fact that  $\sqrt[d]{\cdot}$ is subadditive for any integer $d$.
Finally, we can rewrite Eq.\ \eqref{eq:Pgen2} as
\begin{equation}
\Pgen   \le \frac{1}{2} +2\left( \sqrt[3]{-q} + \sqrt{\frac{p}{3}} \right)+ \delta\;.
\end{equation}
\end{proof}

\noindent
From Theorem \ref{theo:P_guess_fin} we can now rewrite the achievable key rate in \eqref{eq:achievable_keyrate} as
\begin{equation}
\label{eq:achievable_keyrate_final}
    R^{\text{\tiny{QCT}}}_{\infty}  \ge  (1-\abort)\left(-\log \left(\frac{1}{2} +2\left( \sqrt[3]{-q} + \sqrt{\frac{p}{3}} \right)+ \delta \right) - h(\text{QBER})\right)\;.
\end{equation}

\subsubsection{Everlasting secure key expansion}

The security analysis shows that, within the QCT model, we can simplify the scenario to one where Eve's interaction (measurement) on the quantum state occurs right at the beginning, at $t=0$. 
The security analysis after $t_{comp}$, then purely relies  on information-theory principles.
Hence the resulting key rates are valid against an adversary with unbounded computational power after $t_{comp}$, i.e.\  our schemes have everlasting security \cite{unruhEverlastingMultipartyComputation2013}.
We note that everlasting secure key establishment cannot be attained  with cryptographic protocols relying solely on classical communication, even with computational assumptions. Classical communication can be copied, making harvesting attacks (store now, attack later) a significant vulnerability.

Furthermore, to ensure the effectiveness of our hybrid key distribution scheme, the rate of secure key generation must exceed the rate of key consumption due to the need for a pre-shared key. One way to achieve this is by employing a block cipher in the QCT exchange described in Section \ref{sec:protocol_desc}, where Alice divides the message $\vec{y}$ into fixed-size blocks. As a block cipher can encrypt an exponential number of blocks in the key size, the rate of pre-shared key consumption grows logarithmically with the number of protocol rounds, while the final key size increases linearly, ensuring secure key expansion.

\subsection{Performance analysis and trust assumptions}\label{sec:upperbound}
\subsubsection{Key rate analysis}
While most articles on communication complexity focus on asymptotic scalings of input dimensions, our work necessitates precise knowledge of the lower bounds of communication complexity in the non-asymptotic regime. 
\begin{figure}[htb!]
         \centering
         \includegraphics[width=\textwidth]{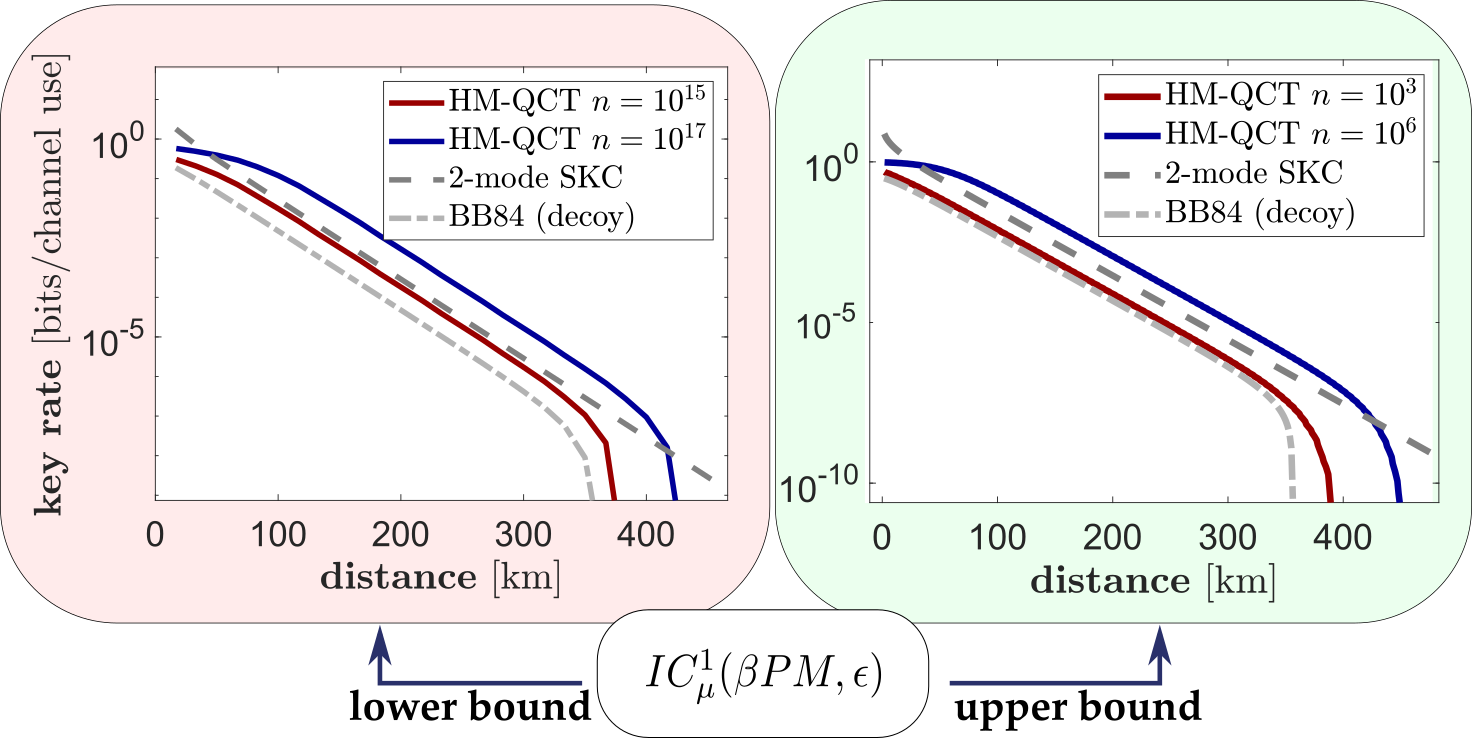}
     \caption[Key rate comparison for the HM-QCT protocol.] {Key rate comparison for the HM-QCT protocol between the derivation from the best-known upper bound and the lower bound of the one-way information complexity of the $\beta PM$ problem,    both evaluated with $\delta= 10^{-4}$ and $\beta = \frac{1}{4}$. We benchmark them with the BB84 protocol with decoy states \cite{loDecoyStateQuantum2005} and the 2-mode  Secret Key Capacity (SKC) \cite{pirandolaFundamentalLimitsRepeaterless2017}. The plots for the HM-QCT protocol are derived under a practical implementation, as detailed in Appendix~\ref{appendix:postsel}. For both the HM-QCT protocol and the BB84 protocol with decoy states, we used the same detector specifications. These detectors are state-of-the-art SNSPDs, as detailed in \cite{liHighrateQuantumKey2023},  characterized by a dark count probability of \( P_{\text{dark}} = 10^{-8} \) and a detection efficiency of $\eta_{\text{det}} = 65\%$.}
\label{fig:keyrate_comparison}
 \end{figure}
With these bounds, we can accurately determine how the guessing probability scales with the dimension and the number of copies sent by Alice, providing a full security proof in an  i.i.d. setting.
Theorem \ref{theo:P_guess_fin} is therefore a significant result, derived from a lower bound of the one-way information complexity of the $\beta\text{PM}$ problem, but this bound may not be tight. In fact, the error $\epsBetaPM(d)$ from the best-known classical protocol with a communication cost $d$ is larger than what one would get from the lower bound.

Nevertheless, one can consider an optimistic scenario where the actual one-way information complexity for any error $\epsilon$ is equal to the information cost of the best-known protocol. This assumes that future developments on finding tighter lower bounds will show that the current best-known classical protocol is the optimal protocol. In particular, in Appendix \ref{appendix:best-known}, we investigate whether alternative approaches to the original best-known protocol exist. This exploration leads to a slightly improved bound, which is presented in Corollary \ref{corol:betapm_communication}.

In this context,  by combining Theorem \ref{theo:boundP_guess} and Lemma \ref{lemma:P_guess_imm} we have 
\begin{equation}
    \Pgen \le 1-\epsBetaPM(m \lceil\log(n)\rceil ) + \delta \;.
\end{equation}

\noindent
Consequently, in Figure \ref{fig:keyrate_comparison} we plot  a comparison between the achievable key rate from \eqref{eq:achievable_keyrate_final} and the key rate obtained from the best-known classical protocol, where in both cases we performed a optimization on the number of copies $m$.

Since our protocol is implemented using two detection modes,  effectively sending at most one bit per channel use, we benchmark it with two standard key rate limits: the BB84  protocol with decoy states \cite{loDecoyStateQuantum2005} and the more general limit for 2-mode optical key distribution \cite{pirandolaFundamentalLimitsRepeaterless2017}, generally called the 2-mode Secret Key Capacity (SKC).
It's evident from the plot that while the lower bound demands an exceedingly high number of modes, around $10^{17}$, to surpass the SKC, the best-known upper bound achieves this with 11 orders of magnitude fewer modes. Moreover, with only a thousand modes, the key rate derived  from the best-known upper bound can already surpass the theoretical limit for BB84. Notably, an experimental implementation of a variant of the  quantum $\beta PM$ protocol has already been performed with a similar number of modes \cite{kumarExperimentalDemonstrationQuantum2019}.

\begin{figure}[htb!]
\centering
\includegraphics[width= 0.65\textwidth]{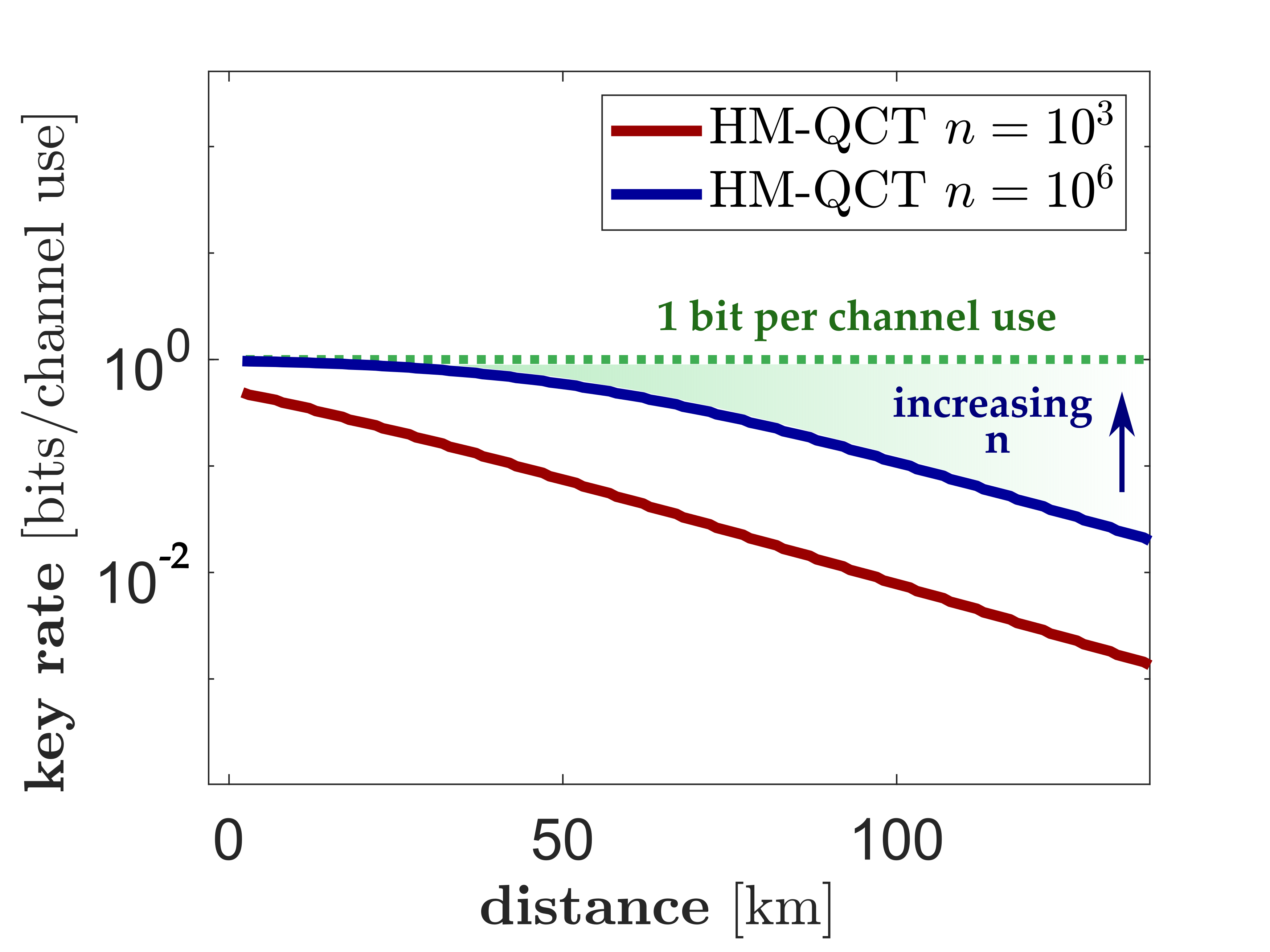}
\caption[Key rate over short distances for the HM-QCT protocol.]{Key rate over short distances for the HM-QCT protocol derived from the best-known upper bound of the one-way information complexity of the $\beta PM$ problem. The specifications of the detectors are the same as those considered in Figure \ref{fig:keyrate_comparison}.
}
\label{fig:keyrate_zoom}
\end{figure}

In general, one can notice that the key rate versus distance behavior of our HM-QCT protocol can be decomposed into essentially three parts:
\begin{itemize}
    \item When focusing on relatively short distances (i.e.\ transmissivity $T \gg \log(n)/\sqrt{n}$), as illustrated in Figure \ref{fig:keyrate_zoom}, a constant rate of 1 secret bit per channel use can be maintained, which corresponds to a classical-like behavior. This is made here possible because the HM-QCT protocol can be operated securely with $\mathcal{O}\left(\sqrt{n}/\log(n)\right)$  input photons.
    
    \item At intermediate distance, the trend of the key rate resembles the 2-mode SKC, decaying exponentially with distance due to Bob receiving, on average, less than one photon per channel use.
    
    \item Close to maximum distance, the probability to detect a signal photon coming from Alice becomes comparable to the detector dark count, provoking the error rate to rise as distance increases and the secret-key rate to drop abruptly.
\end{itemize}

\subsubsection{Semi-device independent approach}
\label{subsec:SDI}

 In our security analysis, one can notice that the way we establish a bound on the min-entropy $H_{\min}(A|Y Z \Phi_{t_{comp}}(Q_{in}))$  only depends on the dimension of the input state prepared by Alice's quantum source. Because of this, it is possible to extend our analysis to the semi-device-independent (SDI) \cite{pawlowskiSemideviceindependentSecurityOneway2011} setting, where we assume that Alice's device generates quantum states of a limited dimension and Bob's measurement device is entirely untrusted. However, the presence of losses in this context may expose quantum cryptographic systems to possible side-channel attacks. 
In a blinding attack \cite{lydersenHackingCommercialQuantum2010}, for example, Eve can exploit the presence of post-selection to remotely control the receiver's device   and execute an intercept-and-resend
attack without raising the QBER.

 \begin{figure}[htb!]
\centering
\includegraphics[width= 0.65\textwidth]{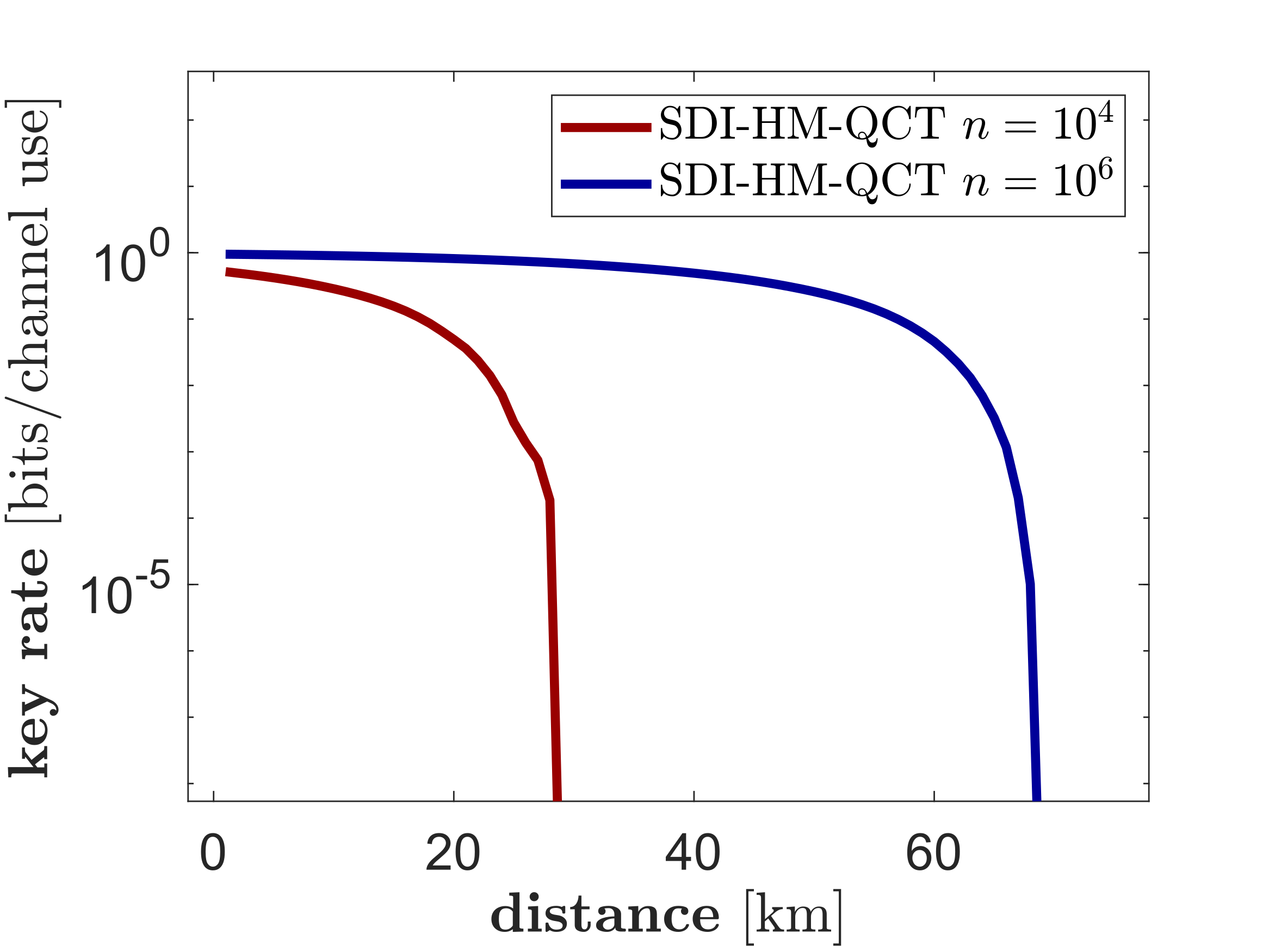}
\caption{Key rate for the SDI-HM-QCT protocol derived from the best-known upper bound of the one-way information complexity of the $\beta PM$ problem. The specifications of the detectors are the same as those considered in Figure \ref{fig:keyrate_comparison}.}
\label{fig:keyrate_SDI}
\end{figure}

Overall, post-selection can cause many problems in the (semi)-device-independent setting \cite{orsucciHowPostselectionAffects2020} and analyses have shown that avoiding these issues places fundamental limits on the achievable key-rates \cite{acinNecessaryDetectionEfficiencies2016}.
Fortunately, due to the high rates achieved by our protocol, we can remove the post-selection and still achieve a positive key rate enabling us to unlock SDI security guarantees.

As shown in Figure \ref{fig:keyrate_SDI}, we can achieve positive key rates  by considering a variant of our protocol called SDI-HM-QCT, which is described in Appendix \ref{appendix:no-postsel}. This variant operates without post-selection, i.e.\ $\abort = 0$, and allows us to derive  positive key rates for short distances from the best-known protocol. 
% Moreover, we believe that one could overcome this short-distance limitation by considering the standard HM-QCT  with post-selection, discussing the critical losses with a general formalism such as the one in \cite{masiniJointmeasurabilityQuantumCommunication2024}.}

\section*{Acknowledgments}
The authors acknowledge financial support from the European Unions’s Horizon Europe research and innovation programme under the project QSNP (Grant No. 101114043) and the PEPR integrated project QCommTestbed (ANR-22-PETQ-0011), which is part of Plan France 2030.
The authors would like to thank the reviewer for identifying an error in the device-independence claims in an earlier version of the manuscript.

\bibliographystyle{quantum}
\bibliography{HM-QCT_bibl1.bib}

\begin{thebibliography}{10}

\bibitem{bennettQuantumCryptographyPublic2014}
Charles~H. Bennett and Gilles Brassard.
\newblock ``Quantum cryptography: {{Public}} key distribution and coin tossing''.
\newblock \href{https://dx.doi.org/10.1016/j.tcs.2014.05.025}{Theoretical Computer Science {\bf 560}, 7--11}~(2014).

\bibitem{pirandolaFundamentalLimitsRepeaterless2017}
Stefano Pirandola, Riccardo Laurenza, Carlo Ottaviani, and Leonardo Banchi.
\newblock ``Fundamental limits of repeaterless quantum communications''.
\newblock \href{https://dx.doi.org/10.1038/ncomms15043}{Nature Communications {\bf 8}, 15043}~(2017).

\bibitem{unruhEverlastingMultipartyComputation2013}
Dominique Unruh.
\newblock ``Everlasting {{Multi-party Computation}}''.
\newblock In Ran Canetti and Juan~A. Garay, editors, Advances in {{Cryptology}} -- {{CRYPTO}} 2013.
\newblock \href{https://dx.doi.org/10.1007/978-3-642-40084-1_22}{Pages 380--397}.
\newblock Lecture {{Notes}} in {{Computer Science}}Berlin, Heidelberg~(2013). Springer.

\bibitem{gavinskyExponentialSeparationsOneway2007}
Dmitry Gavinsky, Julia Kempe, Iordanis Kerenidis, Ran Raz, and Ronald {de Wolf}.
\newblock ``Exponential separations for one-way quantum communication complexity, with applications to cryptography''.
\newblock In Proceedings of the Thirty-Ninth Annual {{ACM}} Symposium on {{Theory}} of Computing.
\newblock \href{https://dx.doi.org/10.1145/1250790.1250866}{Pages 516--525}.
\newblock {{STOC}} '07New York, NY, USA~(2007). Association for Computing Machinery.

\bibitem{bar-yossefExponentialSeparationQuantum2004}
Ziv {Bar-Yossef}, T.s Jayram, and Iordanis Kerenidis.
\newblock ``Exponential separation of quantum and classical one-way communication complexity''.
\newblock In Electronic {{Colloquium}} on {{Computational Complexity}} ({{ECCC}}).
\newblock \href{https://dx.doi.org/10.1145/1007352.1007379}{Pages 128--137}.
\newblock ~(2004).

\bibitem{vyasEverlastingSecureKey2020}
Nilesh Vyas and Romain Alleaume.
\newblock ``Everlasting {{Secure Key Agreement}} with performance beyond {{QKD}} in a {{Quantum Computational Hybrid}} security model''~(2020).
\newblock  \href{http://arxiv.org/abs/2004.10173}{arXiv:2004.10173}.

\bibitem{watrousTheoryQuantumInformation2018a}
John Watrous.
\newblock ``The {{Theory}} of {{Quantum Information}}''.
\newblock \href{https://dx.doi.org/10.1017/9781316848142}{Cambridge University Press}. ~(2018).
\newblock 1 edition.

\bibitem{heshamiQuantumMemoriesEmerging2016}
Khabat Heshami, Duncan~G. England, Peter~C. Humphreys, Philip~J. Bustard, Victor~M. Acosta, Joshua Nunn, and Benjamin~J. Sussman.
\newblock ``Quantum memories: Emerging applications and recent advances''.
\newblock \href{https://dx.doi.org/10.1080/09500340.2016.1148212}{Journal of Modern Optics {\bf 63}, 2005--2028}~(2016).

\bibitem{gidneyHowFactor20482021}
Craig Gidney and Martin Eker{\aa}.
\newblock ``How to factor 2048 bit {{RSA}} integers in 8 hours using 20 million noisy qubits''.
\newblock \href{https://dx.doi.org/10.22331/q-2021-04-15-433}{Quantum {\bf 5}, 433}~(2021).

\bibitem{murtaQuantumConferenceKey2020}
Gl{\'a}ucia Murta, Federico Grasselli, Hermann Kampermann, and Dagmar Bru{\ss}.
\newblock ``Quantum {{Conference Key Agreement}}: {{A Review}}''.
\newblock \href{https://dx.doi.org/10.1002/qute.202000025}{Advanced Quantum Technologies {\bf 3}, 2000025}~(2020).

\bibitem{bravermanInformationEqualsAmortized2011}
Mark Braverman and Anup Rao.
\newblock ``Information {{Equals Amortized Communication}}''~(2011).
\newblock  \href{http://arxiv.org/abs/1106.3595}{arXiv:1106.3595}.

\bibitem{yaoComplexityQuestionsRelated1979}
Andrew Chi-Chih Yao.
\newblock ``Some complexity questions related to distributive computing({{Preliminary Report}})''.
\newblock In Proceedings of the Eleventh Annual {{ACM}} Symposium on {{Theory}} of Computing.
\newblock \href{https://dx.doi.org/10.1145/800135.804414}{Pages 209--213}.
\newblock {{STOC}} '79New York, NY, USA~(1979). Association for Computing Machinery.

\bibitem{kumarExperimentalDemonstrationQuantum2019}
Niraj Kumar, Iordanis Kerenidis, and Eleni Diamanti.
\newblock ``Experimental demonstration of quantum advantage for one-way communication complexity surpassing best-known classical protocol''.
\newblock \href{https://dx.doi.org/10.1038/s41467-019-12139-z}{Nature Communications {\bf 10}, 4152}~(2019).

\bibitem{kumarPracticallyFeasibleRobust2019}
Niraj Kumar.
\newblock ``Practically {{Feasible Robust Quantum Money}} with {{Classical Verification}}''.
\newblock \href{https://dx.doi.org/10.3390/cryptography3040026}{Cryptography {\bf 3}, 26}~(2019).

\bibitem{damgardTightHighOrderEntropic2007}
Ivan~B. Damg{\aa}rd, Serge Fehr, Renato Renner, Louis Salvail, and Christian Schaffner.
\newblock ``A {{Tight High-Order Entropic Quantum Uncertainty Relation}} with {{Applications}}''.
\newblock In Alfred Menezes, editor, Advances in {{Cryptology}} - {{CRYPTO}} 2007.
\newblock \href{https://dx.doi.org/10.1007/978-3-540-74143-5_20}{Pages 360--378}.
\newblock Lecture {{Notes}} in {{Computer Science}}Berlin, Heidelberg~(2007). Springer.

\bibitem{cozzolinoHighDimensionalQuantumCommunication2019}
Daniele Cozzolino, Beatrice Da~Lio, Davide Bacco, and Leif~Katsuo Oxenl{\o}we.
\newblock ``High-{{Dimensional Quantum Communication}}: {{Benefits}}, {{Progress}}, and {{Future Challenges}}''.
\newblock \href{https://dx.doi.org/10.1002/qute.201900038}{Advanced Quantum Technologies {\bf 2}, 1900038}~(2019).

\bibitem{lupoQuantumLockedKeyDistribution2014}
Cosmo Lupo and Seth Lloyd.
\newblock ``Quantum-{{Locked Key Distribution}} at {{Nearly}} the {{Classical Capacity Rate}}''.
\newblock \href{https://dx.doi.org/10.1103/PhysRevLett.113.160502}{Physical Review Letters {\bf 113}, 160502}~(2014).

\bibitem{lumQuantumEnigmaMachine2016}
Daniel~J. Lum, John~C. Howell, M.~S. Allman, Thomas Gerrits, Varun~B. Verma, Sae~Woo Nam, Cosmo Lupo, and Seth Lloyd.
\newblock ``Quantum enigma machine: {{Experimentally}} demonstrating quantum data locking''.
\newblock \href{https://dx.doi.org/10.1103/PhysRevA.94.022315}{Physical Review A {\bf 94}, 022315}~(2016).

\bibitem{lupoContinuousvariableQuantumEnigma2015}
Cosmo Lupo and Seth Lloyd.
\newblock ``Continuous-variable quantum enigma machines for long-distance key distribution''.
\newblock \href{https://dx.doi.org/10.1103/PhysRevA.92.062312}{Physical Review A {\bf 92}, 062312}~(2015).

\bibitem{wehnerCryptographyNoisyStorage2008a}
Stephanie Wehner, Christian Schaffner, and Barbara Terhal.
\newblock ``Cryptography from {{Noisy Storage}}''.
\newblock \href{https://dx.doi.org/10.1103/PhysRevLett.100.220502}{Physical Review Letters {\bf 100}, 220502}~(2008).
\newblock  \href{http://arxiv.org/abs/0711.2895}{arXiv:0711.2895}.

\bibitem{koenigUnconditionalSecurityNoisy2012}
Robert Koenig, Stephanie Wehner, and Juerg Wullschleger.
\newblock ``Unconditional security from noisy quantum storage''.
\newblock \href{https://dx.doi.org/10.1109/TIT.2011.2177772}{IEEE Transactions on Information Theory {\bf 58}, 1962--1984}~(2012).
\newblock  \href{http://arxiv.org/abs/0906.1030}{arXiv:0906.1030}.

\bibitem{santosQuantumObliviousTransfer2022}
Manuel~B. Santos, Paulo Mateus, and Armando~N. Pinto.
\newblock ``Quantum oblivious transfer: A short review''.
\newblock \href{https://dx.doi.org/10.3390/e24070945}{Entropy {\bf 24}, 945}~(2022).
\newblock  \href{http://arxiv.org/abs/2206.03313}{arXiv:2206.03313}.

\bibitem{almeidaBriefReviewQuantum2014}
{\'A}lvaro~J. Almeida, Ricardo Loura, Nikola Paunkovi{\'c}, Nuno~A. Silva, Nelson~J. Muga, Paulo Mateus, Paulo~S. Andr{\'e}, and Armando~N. Pinto.
\newblock ``A brief review on quantum bit commitment''.
\newblock In Second {{International Conference}} on {{Applications}} of {{Optics}} and {{Photonics}}.
\newblock \href{https://dx.doi.org/10.1117/12.2063733}{Volume 9286, pages 189--196}.
\newblock SPIE~(2014).

\bibitem{loWhyQuantumBit1998}
Hoi-Kwong Lo and H.~F. Chau.
\newblock ``Why quantum bit commitment and ideal quantum coin tossing are impossible''.
\newblock \href{https://dx.doi.org/10.1016/S0167-2789(98)00053-0}{Physica D: Nonlinear Phenomena {\bf 120}, 177--187}~(1998).

\bibitem{buhrmanCompleteInsecurityQuantum2012}
Harry Buhrman, Matthias Christandl, and Christian Schaffner.
\newblock ``Complete {{Insecurity}} of {{Quantum Protocols}} for {{Classical Two-Party Computation}}''.
\newblock \href{https://dx.doi.org/10.1103/PhysRevLett.109.160501}{Physical Review Letters {\bf 109}, 160501}~(2012).

\bibitem{ngExperimentalImplementationBit2012}
Nelly Huei~Ying Ng, Siddarth~K. Joshi, Chia Chen~Ming, Christian Kurtsiefer, and Stephanie Wehner.
\newblock ``Experimental implementation of bit commitment in the noisy-storage model''.
\newblock \href{https://dx.doi.org/10.1038/ncomms2268}{Nature Communications {\bf 3}, 1326}~(2012).

\bibitem{ervenExperimentalImplementationOblivious2014}
C.~Erven, N.~Ng, N.~Gigov, R.~Laflamme, S.~Wehner, and G.~Weihs.
\newblock ``An experimental implementation of oblivious transfer in the noisy storage model''.
\newblock \href{https://dx.doi.org/10.1038/ncomms4418}{Nature Communications {\bf 5}, 3418}~(2014).

\bibitem{furrerContinuousvariableProtocolOblivious2018}
Fabian Furrer, Tobias Gehring, Christian Schaffner, Christoph Pacher, Roman Schnabel, and Stephanie Wehner.
\newblock ``Continuous-variable protocol for oblivious transfer in the noisy-storage model''.
\newblock \href{https://dx.doi.org/10.1038/s41467-018-03729-4}{Nature Communications {\bf 9}, 1450}~(2018).

\bibitem{harshaCommunicationComplexityCorrelation2010}
Prahladh Harsha, Rahul Jain, David McAllester, and Jaikumar Radhakrishnan.
\newblock ``The {{Communication Complexity}} of {{Correlation}}''.
\newblock \href{https://dx.doi.org/10.1109/TIT.2009.2034824}{IEEE Transactions on Information Theory {\bf 56}, 438--449}~(2010).

\bibitem{raoCommunicationComplexityApplications2020a}
Anup Rao and Amir Yehudayoff.
\newblock ``Communication {{Complexity}}: And {{Applications}}''.
\newblock \href{https://dx.doi.org/10.1017/9781108671644}{Cambridge University Press}. Cambridge~(2020).

\bibitem{lupoQuantumDataLocking2015a}
Cosmo Lupo.
\newblock ``Quantum {{Data Locking}} for {{Secure Communication}} against an {{Eavesdropper}} with {{Time-Limited Storage}}''.
\newblock \href{https://dx.doi.org/10.3390/e17053194}{Entropy {\bf 17}, 3194--3204}~(2015).

\bibitem{katzIntroductionModernCryptography2014}
Jonathan Katz and Yehuda Lindell.
\newblock ``Introduction to {{Modern Cryptography}}''.
\newblock \href{https://dx.doi.org/10.5555/2700550}{Chapman \& Hall/CRC}. ~(2014).
\newblock 2 edition.

\bibitem{rennerSecurityQuantumKey2008}
Renato Renner.
\newblock ``Security of quantum key distribution''.
\newblock \href{https://dx.doi.org/10.1142/S0219749908003256}{International Journal of Quantum Information {\bf 06}, 1--127}~(2008).

\bibitem{brassardSecretKeyReconciliationPublic1994}
Gilles Brassard and Louis Salvail.
\newblock ``Secret-{{Key Reconciliation}} by {{Public Discussion}}''.
\newblock In Tor Helleseth, editor, Advances in {{Cryptology}} --- {{EUROCRYPT}} '93.
\newblock \href{https://dx.doi.org/10.1007/3-540-48285-7_35}{Pages 410--423}.
\newblock Lecture {{Notes}} in {{Computer Science}}Berlin, Heidelberg~(1994). Springer.

\bibitem{bennettGeneralizedPrivacyAmplification1995}
C.~H. Bennett, G.~Brassard, C.~Cr{\'e}peau, and U.~Maurer.
\newblock ``Generalized privacy amplification''.
\newblock \href{https://dx.doi.org/10.1109/18.476316}{IEEE Trans. Inf. Theory}~(1995).

\bibitem{lutkenhausQuantumKeyDistribution2002}
Norbert L{\"u}tkenhaus and Mika Jahma.
\newblock ``Quantum key distribution with realistic states: Photon-number statistics in the photon-number splitting attack''.
\newblock \href{https://dx.doi.org/10.1088/1367-2630/4/1/344}{New Journal of Physics {\bf 4}, 44}~(2002).

\bibitem{devetakDevetakWinterDistillation2003}
Igor Devetak and Andreas Winter.
\newblock ``Devetak, {{I}}. \& {{Winter}}, {{A}}. {{Distillation}} of secret key and entanglement from quantum states. {{Proc}}. {{R}}. {{Soc}}. {{Lond}}. {{A}} 461, 207-235''.
\newblock \href{https://dx.doi.org/10.1098/rspa.2004.1372}{Proceedings of the Royal Society A: Mathematical, Physical and Engineering Sciences{\bf 461}}~(2003).

\bibitem{gopalUsingPostmeasurementInformation2010}
Deepthi Gopal and Stephanie Wehner.
\newblock ``Using post-measurement information in state discrimination''.
\newblock \href{https://dx.doi.org/10.1103/PhysRevA.82.022326}{Physical Review A {\bf 82}, 022326}~(2010).
\newblock  \href{http://arxiv.org/abs/1003.0716}{arXiv:1003.0716}.

\bibitem{loDecoyStateQuantum2005}
Hoi-Kwong Lo, Xiongfeng Ma, and Kai Chen.
\newblock ``Decoy {{State Quantum Key Distribution}}''.
\newblock \href{https://dx.doi.org/10.1103/PhysRevLett.94.230504}{Physical Review Letters {\bf 94}, 230504}~(2005).

\bibitem{liHighrateQuantumKey2023}
Wei Li, Likang Zhang, Hao Tan, Yichen Lu, Sheng-Kai Liao, Jia Huang, Hao Li, Zhen Wang, Hao-Kun Mao, Bingze Yan, Qiong Li, Yang Liu, Qiang Zhang, Cheng-Zhi Peng, Lixing You, Feihu Xu, and Jian-Wei Pan.
\newblock ``High-rate quantum key distribution exceeding 110 {{Mb}} s--1''.
\newblock \href{https://dx.doi.org/10.1038/s41566-023-01166-4}{Nature Photonics {\bf 17}, 416--421}~(2023).

\bibitem{pawlowskiSemideviceindependentSecurityOneway2011}
Marcin Paw{\l}owski and Nicolas Brunner.
\newblock ``Semi-device-independent security of one-way quantum key distribution''.
\newblock \href{https://dx.doi.org/10.1103/PhysRevA.84.010302}{Physical Review A {\bf 84}, 010302}~(2011).

\bibitem{lydersenHackingCommercialQuantum2010}
Lars Lydersen, Carlos Wiechers, Christoffer Wittmann, Dominique Elser, Johannes Skaar, and Vadim Makarov.
\newblock ``Hacking commercial quantum cryptography systems by tailored bright illumination''.
\newblock \href{https://dx.doi.org/10.1038/nphoton.2010.214}{Nature Photonics {\bf 4}, 686--689}~(2010).

\bibitem{orsucciHowPostselectionAffects2020}
Davide Orsucci, Jean-Daniel Bancal, Nicolas Sangouard, and Pavel Sekatski.
\newblock ``How post-selection affects device-independent claims under the fair sampling assumption''.
\newblock \href{https://dx.doi.org/10.22331/q-2020-03-02-238}{Quantum {\bf 4}, 238}~(2020).

\bibitem{acinNecessaryDetectionEfficiencies2016}
Antonio Ac{\'i}n, Daniel Cavalcanti, Elsa Passaro, Stefano Pironio, and Paul Skrzypczyk.
\newblock ``Necessary detection efficiencies for secure quantum key distribution and bound randomness''.
\newblock \href{https://dx.doi.org/10.1103/PhysRevA.93.012319}{Physical Review A {\bf 93}, 012319}~(2016).

\end{thebibliography}

\newpage

\appendix
\section{Message compression schemes}
\label{appendix:comp_scheme}
In \cite{harshaCommunicationComplexityCorrelation2010}, the authors demonstrated that it is possible to compress each message of a protocol to approximately its contribution to the external information cost plus some additional constant term.
While the authors focused only on the asymptotic scaling, we carefully derived all the specific constants for the compression scheme. 
\begin{theorem}[Message compression]
\label{theo:messagecomp}
Consider a message $M$ sent by Alice, who holds $X$, and where $M$ is sampled from some conditional probability distribution $P_{M|X}$.
Alice and Bob can use public randomness to simulate\footnote{Instead of sending directly $M$, Alice and Bob can use their shared randomness to decrease the number of bits  Alice has to send, while Bob can still retrieve completely the message $M$.} sending  $M$ by sending an expected number of bits
%s\footnote{Note that the expected communication cost is the average version of our worst-case definition of communication cost, i.e.\ Definition \ref{def:communicationcost}.} 
 upper bounded by $I(M:X) +  1.262  \log(1 + I(M:X)) + 11.6$. The simulation is one round  (i.e. only Alice has to send information) and without error.
\end{theorem}

\begin{proof}
First,  we need  to define a one-way compression scheme to transmit integers in an optimal way.
 \begin{lemma}[Compression scheme]
 \label{lem:send_integer}
Let $z$ be an integer. There exists a one-way protocol that allows Alice to  communicate $z$ to Bob using at most
$\log(z) + 1.262\log( \log(z)) + 6.3$ bits.
\end{lemma}
 \begin{proof}
The protocol consists of two phases. In the first phase Alice sends $y := \lceil \log(z)\rceil$ in base $3$ using the two-bit letters $00, 01, 10$. Alice  then sends the bits $11$ to indicate to Bob that the first phase is complete. In the second phase Alice sends the binary representation of $z$ to Bob. Note that because Bob  knows $\lceil\log(z)\rceil$, he knows when the protocol stops.

%\paragraph*{Complexity analysis}

In the first phase of the protocol Alice sends $\lceil\log_3( \lceil \log(z)\rceil) \rceil$ two-bit letters plus an addition two bits to complete the phase. Thus the total number of bits can be bounded as
\begin{align*}
& 2\lceil\log_3( \lceil \log(z)\rceil) \rceil +2 \le 2 \log_3( \lceil \log(z)\rceil) +4 \\
= & 2 \log_3(2) \log( \lceil \log(z)\rceil) +4 
\le  2 \cdot 0.631 \cdot  \log( \lceil \log(z)\rceil) + 4 \\
= & 1.262 \log( \lceil \log(z)\rceil) + 4 
\le  1.262 \log( \log(z) + 1) + 4 \\
\le & 1.262 \log( \log(z) ) + 5.3\;,
\end{align*}
where in the last inequality we used the fact that $\log(x+1)\le \log(x) +1$ for $x\ge1$.
In the second step Alice  only needs to send $\lceil\log( z)\rceil\le \log(z) + 1$ bits. By combining the upper bounds we obtain the claimed result.
\end{proof}
Then, we can use Claim 7.9 of \cite{raoCommunicationComplexityApplications2020a}, and replace the Claim 7.8 by our Lemma \ref{lem:send_integer} to complete the derivation. 
\end{proof}

Finally,  one can then map the one-way distributional communication complexity to the external information complexity in a one-way setting, exploiting   the message compression in Theorem \ref{theo:messagecomp}, obtaining the result in  Lemma V.3 of \cite{harshaCommunicationComplexityCorrelation2010} with explicit constants.
\mappingtoinfo*

\begin{proof}
Let $f$ be a function, $\mu \in \Delta(\mathcal{X} \times\mathcal{Y})$ be a joint probability distribution over the inputs of $f$, and $\pi$ a one-way  private-coin protocol which computes $f$ with error upper bounded by $\epsilon$ such that
\begin{equation}\label{eq:icinf}
IC_\mu^1 (\pi) \le IC^1_\mu(f,\epsilon) + 0.05\;.
\end{equation}
%We define $\mu_X \in \Delta (\mathcal{X})$ the marginal distribution of $\mu$.
We use Theorem \ref{theo:messagecomp} to deduce a new one-way public-coin protocol  $\pi'$   such that  the average size of the transcript is  upper bounded by
%that uses both a shared randomness (public-coin) and private randomness (private coin). As discussed before, since we are eventually interested in the communication cost of $\pi'$, it can be simulated by a public-coin protocol $\pi''$. 
\begin{equation*}
\mathbb{E} \left( |\Pi'| \right)  \le I(\Pi : X) + 1.262 \log(1 + I(\Pi : X))) + 11.6 \;
\end{equation*}
and the  error probability is at most $\epsilon$.
Then we apply the inequality 
$1.262 \log(1+x) \le x + 0.3$ for any $x \ge 0$ to deduce
\begin{align*}
\mathbb{E} \left( |\Pi'| \right) \le & I(\Pi : X) + 1.262 \log(1 + I(\Pi : X))) + 11.6
 \\
\le & 2 I(\Pi : X) + 11.9 \\
\le & 2IC^1_\mu(f,\epsilon) + 12  \;,
\end{align*}
where in the last inequality we used  \eqref{eq:icinf}.
By using Markov's inequality, we can create a new protocol 
$\pi''$ which is identical to $\pi'$ except when the transcript 
$\Pi''$ has size greater than 
$ \frac{1}{\delta_2}\mathbb{E}\left( |\Pi'| \right)$, then the protocol simply aborts. 
By suitably fixing the public randomness, one can a deterministic protocol which has probability to fail upper bounded by 
$\epsilon + \delta_2$ and a communication cost at most $\frac{2IC^1_\mu(f,\epsilon) + 12 }{\delta_2}$. The lemma then follows from Definition \ref{def:distributional} (see Eq.\ \eqref{eq:defin_distibutional}).

\end{proof}

\section{Practical quantum protocols for the $\beta$PM problem}
\label{appendix:practical_prot}
In this section, we examine two variants of a practical quantum protocol designed to address the $\beta PM$ problem. In a typical quantum protocol, Alice transmits $m$ copies of the quantum state $\ket{\psi_x}$ (as defined in Eq.\ \eqref{eq:psi_x}) to Bob. Bob then performs a measurement based on the ideal protocol, which has three possible outcomes:
\begin{enumerate}
    \item   He aborts the protocol with probability $\abort$ if the measurement result is inconclusive ($b = \perp$).
    \item He outputs $b = a $ with probability $ (1 - \abort)(1 - \text{QBER})$.
    \item He outputs $b \neq a$ with probability $ (1 - \abort) \text{QBER}$.
\end{enumerate}

  \noindent
First, we analyze a specific implementation based on photonics. In this setup, each copy of the quantum state is encoded in a photon across $n$ optical modes, and Bob's decision-making process relies on detecting photons using only two detectors. This implementation corresponds to the standard HM-QCT protocol, where Alice and Bob discard all inconclusive rounds.
Following this, we introduce a variant of this photonic implementation that eliminates inconclusive rounds entirely. This approach is essential to  be protected against side-channel attacks that could exploit the presence of post-selection.

% In this section, we analyze two variants of a practical quantum protocol to solve the $\beta PM$ problem.  In a general quantum protocol, Alice sends $m$ copies of the quantum state $\ket{\psi_x}$ in Eq.\ \eqref{eq:psi_x} to Bob. Bob performs a measurement from the ideal protocol where there are three possible outcomes: he aborts the protocol with probability $\abort$ if the measurement result is inconclusive ($b = \perp$); otherwise, he outputs $b = a$ with probability $(1 - \abort)(1 - \text{QBER})$ or $b \neq a$ with probability $(1 - \abort) \text{QBER}$.

% First, we analyze a specific implementation based on photonics, where each copy of the quantum state is encoded in a photon with $n$ optical modes, and where Bob's outcome decision-making process relies on detecting photons using only two detectors. This first implementation is used for the standard HM-QCT protocol, where Alice and Bob discard all the inconclusive rounds. 

% Then we present a variant of this implementation, where there is never an inconclusive round, crucial to be protected against side-channel attacks exploiting the presence of post-selection.  

\subsection{Protocol with post-selection} 
\label{appendix:postsel}
Consider a lossy channel, with $T$ the transmittance of the channel, defined as $T = 10^{-0.02 L}$, where $L$ is the length of the quantum channel expressed in kilometers. Let $\eta_{det}$ be the detector efficiency and $P_{dark}$ the dark-count probability per detector. We will assume that the error rate is dominated by dark counts and that clicks due to signals  and  due to dark counts are independent. 
%Since Bob knows the Hidden Matching instance $y=(M; \omega)$, he can perform the (orthogonal) measurement associated to $M$, with only two detectors. 
In this analysis we will not consider photon counting detectors. Now let us consider the probability of a photon  sent by Alice  being detected: it will be transmitted with probability $T$ due to loss in the transmission channel; once it has successfully reached Bob's measurement apparatus, there is a probability $2\beta$ of addressing one of the modes described by the partial matching; finally, once it is rerouted to one of the two detectors, it will be detected only with probability $\eta_{det}$. Combining all these steps, the final probability  for a photon to be detected is $\tilde{T} \coloneqq 2\beta \eta_{det} T $.  Since each photon is independent, the probability that there is at least one click due to the signal is

\begin{equation}
\label{eq:P_s}
    P_s= 1-(1-\tilde{T})^m\,.
\end{equation}
 Hence, compared to standard QKD protocols that only allow to send one copy per channel use, sending multiple photons per channel use exponentially increases the probability of obtaining a click, thereby improving overall performance.

Moreover, the probability of getting zero clicks is the probability of having at the same time no clicks from dark counts and no clicks due to the actual signal, i.e.
 \begin{align*}
\text{Pr}(\text{0 clicks})&= \text{Pr}(\text{no dark counts}) \cdot\text{Pr}(\text{0 clicks due to signal})\\
&= (1-2P_{dark}+ P_{dark}^2) (1-P_s)\;.
\end{align*}
\noindent
On the other hand, the probability of getting a click in both detectors at the same time is
\begin{align*}
\text{Pr}(\text{both detectors click})\!&= \!\text{Pr}(\text{2 dark counts}) \!+\! \text{Pr}(\text{dark count in wrong detector}) \cdot P_s \\
&= P_{dark}^2 +P_{dark}(1-P_{dark}) P_s\;.
 \end{align*}
\noindent
We now assume that Bob aborts the protocol every time he has $0$ clicks or clicks in both detectors, obtaining
\begin{align}
\abort &= (1-2P_{dark}+ P_{dark}^2) (1-P_s) + ( P_{dark}^2 +P_{dark}(1-P_{dark}) P_s)\nonumber\\
& =  P_{dark} +(1 -3P_{dark} + 2 P_{dark}^2)(1-\tilde{T})^m \;.
\label{eq:p(succ)}
\end{align}
\noindent
Now we have  that  the QBER is the probability of giving a wrong answer after the sifting:
\begin{equation}
\text{QBER}= \frac{\text{Pr}(b\neq a \land b\neq \perp)}{1-\abort}\;,
\end{equation}
with $\text{Pr}(b\neq a \land b\neq \perp)= P_{dark} (1-P_{dark})(1-P_s) $, obtaining eventually by direct calculation 
\begin{equation}
\text{QBER} =\frac{P_{dark} - P_{dark}^2}{1- P_{dark} -(1 -3P_{dark} + 2 P_{dark}^2)(1-\tilde{T})^m } (1-\tilde{T})^m \;.
\label{eq:QBER}
\end{equation}
Finally, we have evaluated the $\abort$\footnote{One can notice that even in the case where Alice is sending a large number of copies $\abort$ converges to $P_{dark}$ instead of simply $0$. This is due to the fact that we haven't considered an implementation with photon counting detectors.} and QBER as a function of the number of copies sent $m$. 

\subsection{Protocol without post-selection}
\label{appendix:no-postsel}
In this variant of the protocol, we assume that whenever Bob records 0 clicks or clicks in both detectors, he outputs the bit $b$ randomly, without aborting the protocol. This implies that $\abort = 0$, and consequently, Eq.\ \eqref{eq:QBER} is modified to become\vspace{-0.2cm}
\begin{align}
    \text{QBER} &=   P_{dark} (1-P_{dark})(1-P_s) + \frac{1}{2}\left( \text{Pr}(\text{0 clicks}) + \text{Pr}(\text{both detectors click})\right) \nonumber\\
   & = \frac{1}{2} \left(\left(1-P_{dark}  \right) (1-\tilde{T})^m +P_{dark}\right)\;.
\end{align}
Here, it is important to emphasize that the ability to send multiple photons per channel use enhances once again the overall performance of the protocol, this time by allowing for an exponential decrease in the QBER.
% \begin{align}
%     \text{QBER} &=   P_{dark} (1-P_{dark})(1-P_s) + \frac{1}{2}( \text{Pr}(\text{0 clicks}) + \text{Pr}(\text{both detectors click}))\\
%     & = \left(P_{dark} - P_{dark}^2 +\frac{1}{2} - \frac{3}{2}P_{dark} + P_{dark}^2 \right) (1-\tilde{T})^m + \frac{1}{2}P_{dark}\\
%     & = \left(\frac{1}{2} - \frac{1}{2}P_{dark}  \right) (1-\tilde{T})^m + \frac{1}{2}P_{dark}\\
%    & = \frac{1}{2} \left(\left(1-P_{dark}  \right) (1-\tilde{T})^m +P_{dark}\right)
% \end{align}

\section{Derivation of Theorem \ref{theo:exponential_sep}}
\label{appendix:deriv_th}

In \cite{gavinskyExponentialSeparationsOneway2007} the authors prove that, given $\beta \in (0, 1/4]$\footnote{Note that in this work we have used the notation $\beta$ in place of the $\alpha$ from \cite{gavinskyExponentialSeparationsOneway2007}.} and $\epsilon_1 \in (0,1/2)$, for any deterministic protocol $\pi$ for the $\beta$-partial Matching Problem that  has a communication cost at most $\gamma\epsilon_1 
\sqrt{n/\beta} + \log(\epsilon_1)$, with $\gamma$ a positive constant which we will determine afterwards, the probability of success with respect to the distribution $\mu$  is upper bounded by $\frac{1}{2} + \frac{5}{2}\sqrt{\epsilon_1}$. 
To make the correspondance with  Theorem 2, we  can  write $\epsilon_1$ in terms of the error probability $\epsilon$ by noticing that
$1 - \epsilon   \le \frac{1}{2} + \frac{5}{2}\sqrt{\epsilon_1}$. This in fact implies $\epsilon_1 \ge \frac{4}{25}\left(\frac{1}{2} - \epsilon\right)^2$.
By definition of the distributional complexity we can therefore obtain Theorem 2, where all we need now is to  retrieve  the desired upper bound for $\gamma$.

\subsection{About \texorpdfstring{$\gamma$}{TEXT}}\label{sec:gamma}

Still from \cite{gavinskyExponentialSeparationsOneway2007}, in their analysis they require the value of $\gamma$ to be small enough to satisfy the following inequalities:

\begin{align}
\frac{\epsilon_1^2}{2} \ge & \sum^{4c -2}_{\text{even }k=2} \left(\frac{64e\gamma^2\epsilon_1^2}{k} \right)^{k/2}  \label{eq:gammasomme}\\
\frac{\epsilon_1^2}{2} \ge & \left(8\sqrt{2}e\gamma\epsilon_1\sqrt{\frac{\beta}{n}}\right)^{2c} \label{eq:gamma2}\;,
\end{align}
\noindent
with $c\ge 1$. First, let's  prove that the bound $\gamma \le \frac{1}{8e}$ implies  Eq.\ \eqref{eq:gammasomme}.
We notice that $\gamma \le \frac{1}{8e} \le  \sqrt{\frac{1}{96e}},$
 resulting  in $ 96e\gamma^2  \le  1$. Then we obtain the following bound for $\frac{\epsilon_1^2}{2}$:

\begin{align*}
\frac{\epsilon_1^2}{2} &\ge \frac{96e\gamma^2\epsilon_1^2 }{2} && \text{(From  $ 1\ge 96e\gamma^2  $) } \\
% &\ge \frac{96e\gamma^2\epsilon_1^2 }{3-96e\gamma^2} && \text{(Explanation for the second inequality)} \\
&\ge \frac{32e\gamma^2\epsilon_1^2 }{1-32e\gamma^2} && \text{(Using  $ 2\le 3-96e\gamma^2  $) } \\
&\ge \sum^{\infty}_{k=1} \left(32e\gamma^2 \right)^{k} \epsilon_1^2 && \text{(Given  $\sum_{k=1}^{\infty}x^k = \frac{x}{1-x}$)} \\
&\ge \sum^{\infty}_{k=1} \left(32e\gamma^2 \epsilon_1^2\right)^{k} && \text{(From $\epsilon_1 < 1$)} \\
% &\ge \sum^{\infty}_{k=1} \left(\frac{64e\gamma^2\epsilon_1^2}{2} \right)^{k} && \text{(Explanation for the sixth inequality)} \\
&\ge \sum^{\infty}_{\text{even }k=2} \left(\frac{64e\gamma^2\epsilon_1^2}{k} \right)^{k/2} && \text{(Using $k>1$)} \\
&\ge \sum^{4c -2}_{\text{even }k=2} \left(\frac{64e\gamma^2\epsilon_1^2}{k} \right)^{k/2} \;. && \text{(Truncating the sum).}
\end{align*}

To conclude, we demonstrate that \(\gamma \le \frac{1}{8e}\) implies \eqref{eq:gamma2}. First, we notice that we can rewrite the bound as $\frac{1}{2} \ge \left(4\sqrt{2}e\gamma\right)^{2} $. Then, as before, we derive the desired upper bound for $\frac{\epsilon_1^2}{2}$:
\begin{align*}
\frac{\epsilon_1^2}{2} &\ge \left(8\sqrt{2}e\gamma\epsilon_1\frac{1}{2}\right)^{2} && \text{(From $\frac{1}{2} \ge \left(4\sqrt{2}e\gamma\right)^{2} $)} \\
&\ge \left(8\sqrt{2}e\gamma\epsilon_1\sqrt{\frac{\beta}{n}}\right)^{2} && \text{(Using \(\epsilon_1 <\frac{1}{2}\))} \\
&\ge \left(8\sqrt{2}e\gamma\epsilon_1\sqrt{\frac{\beta}{n}}\right)^{2c} \;, && \text{(Given \(c\ge 1\) and \(\beta/n \le \frac{1}{4}\))}\;.
\end{align*}

\section{Best-known classical protocol}
\label{appendix:best-known}
In this section we analyze the original best-known classical protocol for the $\beta PM$ problem, called $\piOrig$, which has already been sketched in \cite{gavinskyExponentialSeparationsOneway2007, bar-yossefExponentialSeparationQuantum2004}\footnote{These are the only articles that we know of that investigate this particular problem.}. Moreover, we will discuss possible natural alternatives, converging to a slightly better version.
%Notably, one of these suboptimal strategies reflects perfectly a realistic attack by Eve to our HM-QCT protocol.
\\[0.2cm]
\noindent
\textbf{Original classical protocol $\bm{\piOrig}$:} Alice and Bob can exploit their public randomness to agree on a subset  $\s\! \coloneqq \!\{j_1,\dots,j_d\}$ $\in \S_{d}$, where $\mathcal{S}_d$ is the set of all the possible subsets of $d$ indices in $[n]$.  Subsequently, Alice transmits the corresponding bit values $x_{\s} \coloneqq (x_{j_1}, x_{j_2},\dots x_{j_d})$ to Bob. As such, the communication cost of this protocol is  $d$.  Consequently, in this protocol, Bob receives the corresponding $n_{\text{edges}}(d)=\frac{d(d-1)}{2}$ edges\footnote{Whenever we say that Bob receives an edge, say $(j_1,j_2)$,  it implies that he acquires the bit values assigned to the corresponding vertices, i.e.\ $(x_{j_1}, x_{j_2})$.}. We call $\sigma(\s)$ the set of all those edges. Finally, Bob, by knowing $\omega$, can give the right answer whenever he gets at least an edge in the  matching $M$ and randomly guesses the bit otherwise.\\[0.2cm]
\noindent
From our analysis, we find an upper  bound of the error probability:

\begin{theorem}
\label{theo:betapm_communication}
Let $d$ be an integer. An explicit one-way public-coin protocol $\piOrig$ exists with a communication cost $CC(\piOrig)= d$  which solves the $n$-dimensional $\beta PM$ protocol with an error probability for any input at most\footnote{Note that in \eqref{eq:eps_betaPM} we considered $\binom{a}{b} = 0$ whenever $b>a$.}
 \begin{equation}
 \epsOrig(d)=  \sum_{k=0}^{d} \frac{\binom{2 \beta n}{k} \binom{n-2 \beta n}{d-k}}{2\binom{n}{d}} e^{-\frac{k(k-1)}{4\beta n}}\;.
 \label{eq:eps_betaPM}
 \end{equation}
\end{theorem}

 \begin{proof}
 First, we define $\Mlist$ as the list of all the vertices in the $\beta$-matching $M$.  For example, let $n=4$ and $M$ be a perfect matching (i.e.\ $\beta =1/2$) such that $M = \{(1,2),(3,4)\}$, then $\Mlist = \{1,2,3,4\}$.
We  call $d_{M}$  the  number of indices in  $\s$  that are part of $\Mlist$, i.e.  $d_{M} \coloneqq \lvert\s \cap \Mlist\rvert$.
One can evaluate  probability distribution of $d_{M}$:
\begin{equation}
\label{eq:d_beta}
P\left( d_{M}=k \right) = \frac{\binom{2 \beta n}{k} \binom{n-2 \beta n}{d-k}}{\binom{n}{d}}\;,
\end{equation}
 where  $\binom{n}{d}$ is the number of ways to pick $d$ indices in $[n]$,  $\binom{2 \beta n}{k}$ is the number of ways to pick $k$ indices  which are part of a $\beta$-matching $\Mlist$  and $ \binom{n-2 \beta n}{d-k}$ is instead the number of ways to pick $d-k$ indices which are not part of a $\beta$-matching $M$. 

 We now want to evaluate the probability of Bob not receiving any edge which is part of his $\beta$-matching for a known value of $d_{M}$.
Trivially, whenever Bob doesn't receive any index in $\Mlist$ then the probability of not receiving any edge which is part of $M$, i.e.\ $d_{M}=0$,  is  always equal to $1$, otherwise we have
\begin{equation}
\begin{aligned}
P(\nexists \e \in \sigma(\s)\, \, \text{s.t.}\,\,\e \in M|d_{M} = k) &= \prod_{l=1}^{k} \left(\frac{2 \beta n - 2(l-1)}{2 \beta n - (l-1)}  \right)\\
& = \prod_{l'=0}^{k-1} \left(1- \frac{l'}{2 \beta n - l'}\right) \\
&\le \prod_{l'=0}^{k-1} \left(1- \frac{l'}{2 \beta n}  \right) \\
& \le e^{-\sum_{l'=0}^{k-1} \frac{l'}{2\beta n}} \\
& \le e^{-\frac{k(k-1)}{4\beta n}}\;,
\end{aligned}
\label{eq:PE_l}
\end{equation}
where in the first line we used that, after having checked that the first $l-1$ indices in $\s_{d_{M}}$  do not form any edge in $M$, $2 \beta n - 2(l-1)$ is the remaining number of possible indices in $\Mlist$ that won't form an edge  in $M$ when paired with the indices in the already extracted list $\{j'_1, \dots j'_{l-1}\}$, and  $2 \beta n - (l-1)$ is the  total number of remaining indices in $\Mlist$. In the second line we have simply replaced $l$ with $l' \coloneqq l-1$. The third line is obtained by noticing that $\frac{a}{x-a} > \frac{a}{x}$ for any $x,\,a>0$ with $x>a$. The fourth and fifth lines come from $ 1-x < e^{-x}$ and $\sum_{i=0}^{k-1} i= k(k-1)/2$ respectively. 

Finally, since Bob, by knowing $\omega$, can give the right answer whenever he gets at least an edge in the  matching $M$ and randomly guesses the bit otherwise, the error probability  for the best-known protocol  is at most
\begin{align*}
\frac{1}{2} \sum_{k=0}^{d}  P\left( d_{M}=k \right)   & P(\nexists \e \in \sigma(\s) \, \, \text{s.t.}\,\,\e \in M|d_{M} = k)\\
& \le  \frac{1}{2} \sum_{k=0}^{d}  P\left( d_{M}=k \right)e^{-\frac{k(k-1)}{4\beta n}}    \\
& \le  \sum_{k=0}^{d} \frac{\binom{2 \beta n}{k} \binom{n-2 \beta n}{d-k}}{2\binom{n}{d}} e^{-\frac{k(k-1)}{4\beta n}}\;.
\end{align*}
where in the second line we used the fact that $d_{M}$ cannot be larger than $d$, Eq.\ \eqref{eq:PE_l} in the third line and \eqref{eq:d_beta} in the last line.
 \end{proof}
 One might consider varying the probability distribution for the subset $\bm{s}$ of indices that Alice sends to Bob, leveraging the knowledge that Bob's input is an edge-disjoint matching with $\beta n$ edges. However, due to the specific structure of Bob's input, the probability of any random edge $(i,j)$ being part of Bob's matching $M$ is uniform. The uniformity of this distribution indicates that modifying the probability distribution of the shared subset \( \bm{s} \) will not lead to a more optimized protocol.

Another alternative would be to change the type of information that is shared: instead of Alice sending the values of $d$ random vertices, she could directly send the values of $d$ edges.  However, choosing these edges randomly results in a suboptimal choice. This is because edge values are not independent; if you have two contiguous edges, say $(i_1, i_2)$ and $(i_2, i_3)$, you can infer the edge value of $(i_1, i_3)$ as $x_1 \oplus x_3 = (x_1 \oplus x_2) \oplus (x_2 \oplus x_3)$. The optimal choice would be to send only sets of $d$ edges that form open paths. In this way, one would effectively send the information about $n_{\text{edges}}(d) = \frac{d(d+1)}{2}$ edges. 
\\[0.2cm]
\noindent
\textbf{Classical protocol $\bm{\piBetaPM}$:} Alice and Bob can exploit their public randomness to agree on a subset  $\s\! \coloneqq \!\{j_1,\dots,j_{d+1}\}$ $\in \S_{d+1}$.  Subsequently, Alice transmits the corresponding $d$ edge values $e_{\s} \coloneqq (x_{j_1}\oplus x_{j_2}, x_{j_2}\oplus x_{j_3},\dots x_{j_d}\oplus x_{j_{d+1}})$ to Bob. As such, the communication cost of this protocol is  $d$. Since  the edges $e_{\s}$ form an open path, Bob can effectively infer  the value of a total of  $n_{\text{edges}(d)=}\frac{d(d+1)}{2}$ edges. Finally, as in the original protocol, Bob, by knowing $\omega$, can give the right answer whenever he gets at least an edge in the  matching $M$ and randomly guesses the bit otherwise.\\[0.2cm]

\noindent
This is evidently a slightly better protocol than the original, resulting in an error probability for a communication cost $d$ equivalent to the one obtained in $\piOrig$  for a communication cost $d+1$.
\begin{corollary}
\label{corol:betapm_communication}
Let $d$ be an integer. An explicit one-way public-coin protocol $\piBetaPM$ exists with a communication cost $CC(\piBetaPM)= d$  which solves the $n$-dimensional $\beta PM$ protocol with an error probability for any input at most
 \begin{equation}
 \begin{aligned}
 \epsBetaPM(d)&=\epsOrig(d+1)\\
 &=\sum_{k=0}^{d+1} \frac{\binom{2 \beta n}{k} \binom{n-2 \beta n}{d+1-k}}{2\binom{n}{d+1}} e^{-\frac{k(k-1)}{4\beta n}}\;.
 \label{eq:eps_betaPM2}
 \end{aligned}
 \end{equation}
\end{corollary}
\noindent

To conclude, since the ultimate goal of solving the $\beta PM$ problem is to share the parity value of at least one of Bob's $\beta n$ edges, of which Alice has no information, it is clear that an optimal strategy will aim to maximize the number of edges shared at a fixed communication cost. This is why we believe that randomly selecting edges that form an open path serves as a natural heuristic for developing an efficient protocol and currently represents the best-performing approach.

%
%\begin{equation}
%\Pgen \le \frac{1}{2} +2^{1 - \frac{1}{3}}\left(\sqrt[3]{-q + \sqrt{\frac{-\Delta}{27}}} + 
%\sqrt[3]{-q - \sqrt{\frac{-\Delta}{27}}}\right)+ \delta\;.
%\label{eq:P_gen_bound_proof}
%\end{equation}
%Now we prove that the second  term with $\sqrt[3]{\cdot}$ is negative by noticing that $-\Delta = 4p^3 + 27q^2 > 27q^2 $, resulting in $\sqrt{-\frac{\Delta}{27}} > \sqrt{q^2} =|q|=-q$. 
%%This implies that 
%%\begin{equation}
%% \sqrt[3]{-q - \sqrt{-\frac{\Delta}{27}}} <0 \;. 
%% \label{eq:negative_term}
%%\end{equation}
%Now we can  further upper bound the guessing probability in \eqref{eq:P_gen_bound_proof}
%
%\begin{align*}
%\Pgen & \le \frac{1}{2} +2^{1 - \frac{1}{3}}\left(\sqrt[3]{-q + \sqrt{\frac{(4p^3 + 27q^2)}{27}}} \right)+ \delta\\
%& \le \frac{1}{2} +2\left(\sqrt[3]{-q} + \sqrt{\frac{p}{3}} \right)+ \delta\;,
%\end{align*}
%where in the  last inequality we used the fact that $\sqrt[d]{\cdot}$ is a  subadditive function for any integer $d$ and that $q<0$.
%

\end{document}